\documentclass[11pt,a4paper]{article}

\usepackage{amsmath,amssymb,amsthm} \usepackage{amsfonts} \usepackage{multicol}
\usepackage{interval}

\usepackage[american]{babel}

\usepackage{csquotes,newcent,textcomp}

\usepackage{enumerate}

\theoremstyle{plain}
\newtheorem{theorem}{Theorem}[section]
\newtheorem{lemma}[theorem]{Lemma}
\newtheorem{prop}[theorem]{Proposition}
\newtheorem{corollary}[theorem]{Corollary}
\newtheorem{defn}[theorem]{Definition}
\newtheorem{remark}[theorem]{Remark}

\title{A framework for the valuation of insurance liabilities by production cost}
\author{Christoph M\"ohr \\
\and Independent researcher, christoph.moehr@gmail.com}
\date{Second version, 1 June 2025}

\begin{document}

\maketitle

\begin{abstract}
This paper sets out a framework for the valuation of liabilities from insurance contracts that is intended to be economically realistic, elementary, reasonably practically applicable, and as a special case to provide a basis for the valuation in regulatory solvency systems such as Solvency II and the SST. The valuation framework is based on the cost of producing the liabilities to an insurance company that is subject to solvency regulation (regulatory solvency capital requirements) and insolvency laws (definition and consequences of failure) in finite discrete time. Starting from the replication approach of classical no-arbitrage theory, the framework additionally considers the nature and cost of capital (expressed by a ``financiability condition"), that the liabilities may only be required to be fulfilled ``in sufficiently many cases" (expressed by a ``fulfillment condition"), production using ``fully illiquid" assets in addition to tradables, and the asymmetry between assets and liabilities. We identify the necessary and sufficient property of the financiability condition such that the framework recovers the market prices of tradables. We investigate extending production to take account of insolvency and implications of using illiquid assets in the production. We show how Solvency II and SST valuation can be derived with specific assumptions.  

\end{abstract}

\vspace{10pt}


\newpage

\section{Introduction}\label{sec-intro}

Insurance contracts can be viewed as arrangements under which an insurance company promises the payment of future claims and benefits to policyholders and beneficiaries, for premiums paid by the policyholders. Characteristically, the amounts of future claims and benefits are currently uncertain, whereas the premiums are deterministic or at least ``less uncertain''. Claims and benefits may for example depend on the damage caused by natural perils such as windstorms or floods, the death of a person, or the development of financial markets. Insurance contracts thus correspond to stochastic future cash flows, with cash inflows to the insurance company for premiums, and cash outflows for claims and benefits. In this paper, we assume that the stochastic cash flows are given and consider finite discrete time.

The objective of the paper is to set out an elementary and reasonably practically applicable framework for the valuation of insurance contracts in an economically realistic setting. All insurance contracts of an insurance company are valued jointly. The framework is intended as a special case to provide a basis for valuation approaches in regulatory solvency systems such as Solvency II (\cite{EUFD2009}, \cite{EUDR2014}) and the Swiss Solvency Test (SST, \cite{AVOSST}). In the framework, the value is based on the cost to the insurance company of fulfilling the promises made under the insurance contracts, i.e. of producing the insurance contract cash outflows (production cost). The insurance company produces the insurance contract cash outflows according to a production strategy. Under the production strategy, it dynamically trades assets and produces the insurance contract cash outflows by selling assets or paying out cash inflows such as premiums, bond coupons, or equity dividends. 

This concept of production strategies is similar to replication or hedging strategies considered in classical no-arbitrage theory. To capture an economically realistic setting, we consider the following four additional aspects: 

(1) Insurance companies may not be able to fulfill their insurance obligations, i.e. the promises made under the insurance contracts, with certainty, i.e. ``in every possible state of the world''. For example, they may not be able to fulfill the maximum of all their claims and benefits.\footnote{Diversification is a core element of the business model of insurance companies, which should make it unlikely that all future claims and benefits would simultaneously realize at their maximum.} In turn, insurance companies are subject to solvency regulation and supervision and to insolvency laws. Insolvency laws determine when a company fails - typically, when it defaults on an obligation that is due or becomes balance sheet insolvent - and the consequences of failure, for example bankruptcy liquidation. Solvency regulation can be interpreted as intending to make failure sufficiently unlikely. Solvency regulation typically includes regulatory solvency capital requirements such as under Solvency II and in the SST. These can be interpreted as requiring that insurance companies demonstrate, based on a model of stochastic future developments, that they are able to ``sufficiently fulfill'' their insurance obligations, i.e. to produce them ``in sufficiently many cases'' as defined by what we call a \emph{fulfillment condition} in this paper. Under Solvency II, the fulfillment condition corresponds to a 99.5\% value-at-risk and, in the SST, to a 99\% tail-value-at-risk (expected shortfall), in both cases over a one-year period. The production strategies that we consider thus in general provide sufficient as opposed to ``full'' fulfillment, with a fulfillment condition that is required to be satisfied over each one-year period, and where, in case of failure, insolvency laws trigger insolvency proceedings.\footnote{In reality, regulatory solvency capital requirements and insolvency laws are not always sufficiently aligned for sufficient fulfillment to correspond to no failure in ``sufficiently many cases''.} 

(2) Insurance companies have and use capital, for example shareholder capital, which should thus be an element of production strategies. We view capital investments as arrangements by which investors (capital providers) provide financial resources to an insurance company in exchange for promises of future benefits that depend on the outcome and cannot by themselves trigger failure.\footnote{Hybrid capital, more generally, can trigger failure in specific situations.} This allows capital investments to absorb risk, i.e. volatility in the insurance liabilities: if insurance liabilities turn out to be higher, benefits to capital investors can be lower. We assume that investors are prepared to enter into a capital investment if the stochastic future benefits satisfy a specific \emph{financiability condition}. In line with the conception and the one-year view of Solvency II and the SST, we assume that the financiability condition is expressed as a one-year condition, and make the strong assumption that an insurance company can raise capital annually if required provided that the financiability condition is satisfied. 

(3) The assets of insurance companies are typically not limited to tradables, that is, ``fully liquid'' assets, but include also illiquid assets, for which trading is limited. We specifically consider ``fully illiquid assets'', which we assume can neither be bought nor sold.\footnote{``Fully illiquid assets'' in this sense are linked to the question whether insurance companies can benefit from ``holding assets to maturity'' (HTM) instead of having to sell them for a potentially unfavorable market price, which under Solvency II is related to ``volatility adjustment'' and ``matching adjustment''.} This implies that their use in production strategies is limited to using their cash inflows. 

(4) We distinguish between production strategies that only allow for non-negative units of tradables and production strategies that also allow for negative units. The latter amounts to extending the range of available instruments by allowing to take on additional liabilities by borrowing tradables (``going short''). This requires precisely defining the additional liabilities and to consider that they can trigger insolvency. This turns out to be linked to the question whether cash inflows at future dates, specifically from fully illiquid assets and insurance contract premiums, can be used for producing insurance contract cash outflows at earlier dates. Arguably, regulatory solvency capital requirements in some cases implicitly assume that this is possible.\footnote{For example, by acknowledging that the requirements account for ``solvency'' but not ``liquidity''.} As we show, this is related to whether production cost can be negative, e.g. due to premium inflows significantly exceeding cash outflows. 

A key element of the framework is the precise definition of production strategies that accounts for these aspects. We assume that production strategies proceed over successive one-year periods, in line with the one-year view for example under Solvency II and in the SST.  We initially consider production strategies with only non-negative units of tradables and that stop when failure occurs. Restricting to non-negative units of tradables is also worthwhile given that insurance companies often do not take on significant additional liabilities. The strategies can be extended under suitable assumptions to production strategies that can also contain negative units of tradables, and given suitable insolvency laws to production strategies that continue after failure.

One of the main results of this paper (Theorem~\ref{prop-valueextension}) characterizes when the valuation according to the framework recovers the market prices of tradables, more precisely, the market values of non-negative investment strategies of tradables. This is shown for a setting that corresponds to classical no-arbitrage theory, given by production strategies that can contain negative units of tradables, no fully illiquid assets, and the ``full'' fulfillment condition (i.e. requiring fulfillment almost surely), but additionally allowing for capital investments. Theorem~\ref{prop-valueextension} in particular sets a requirement on the financiability condition. For more general settings, specifically allowing for only sufficient fulfillment, we would not expect that the value would be recovered. 

A key idea of the paper is a two step procedure for constructing production strategies recursively backward in time over successive one-year periods. We illustrate this by a very simple concrete example. The example illustrates basic ideas of the earlier paper M\"ohr \cite{Moehr2011} that the framework of this paper shares but attempts to improve and generalizes. A similar setting is used in Section~\ref{ss-valuesolvencyiisst}, where we consider Solvency II and the SST as special cases. For the example, consider the task of calculating the production cost $\Bar{v}_0^{\phi}(L_1)$ at date $0$ of a production strategy $\phi$ for a given liability $L_1\geq 0$ at date $1$, with no cash flows between $0$ and $1$. In general, $L_1$ could be the production cost at date $t=1$ derived in the prior backward recursion. Let $r_{0,1}$ be the risk-free rate at date $0$ for a deterministic payoff at date $1$. As fulfillment condition, we assume that the liabilities $L_1$ need to be fulfilled by the assets $A_1$ with at least $99.5\%$ probability, which can be expressed by $\rho(A_1-L_1)\leq 0$, where $\rho$ is the $99.5\%$ value-at-risk (Section~\ref{ss-valuesolvencyiisst}). 

The first of the two steps of the construction is to determine assets with value $A_0$ at date $t=0$ and a strategy $\phi$ until $t=1$ such that the fulfillment condition is satisfied at $t=1$. For the example, let the strategy $\phi$ consist in investing tradables with value $A_0\geq 0$ at date $0$ risk-free until $t=1$, so that $A_1=(1+r_{0,1})\cdot A_0$. Then, if we select $A_0=(1+r_{0,1})^{-1}\cdot \rho(-L_1)$, the fulfillment condition is satisfied with equality. This is because then $A_1=\rho(-L_1)$ and thus using translation-invariance of $\rho$, we get $\rho(A_1-L_1)=-\rho(-L_1)+\rho(-L_1)=0$. The second step consists in using the financiability condition to split the value $A_0$ into a (maximum) amount $C_0$ that can be raised from capital investors and the remainder, which, as we show, is the production cost $\Bar{v}_0^{\phi}(L_1)=A_0-C_0$ under the selected strategy $\phi$. As is explained in Section~\ref{ss-nonnegprodstrategydefn}, the capital investors' claim at date $1$ is $(A_1-L_1)_{+}$. If this non-negative quantity is not positive with positive probability, no capital is used, $C_0=0$. Otherwise, a capital investor would likely be prepared to invest some positive amount $C_0>0$ at date $0$ for the claim $(A_1-L_1)_{+}$ at $t=1$, depending on the specific financiability condition. The production cost then is $\Bar{v}_0^{\phi}(L_1)=A_0-C_0$ for the following reason: given assets with value $A_0-C_0$, capital $C_0$ can be raised, giving assets with value $A_0$ that are then invested in the strategy. By construction, the financiability and the fulfillment condition then hold, so the liability $L_1$ can be produced. 

Further simplifying for illustration, let the financiability condition be given by a minimum expected excess return $\eta>0$ over the risk-free rate $r_{0,1}$, where we also ignore $()_{+}$, i.e. we have as minimum $E[A_1-L_1]= (1+r_{0,1}+\eta)\cdot C_0$. It then follows from a short calculation that the production cost under the strategy can be written as
\begin{equation}
	\Bar{v}_0^{\phi}(L_1)=\frac{E[L_1]}{1+r_{0,1}}+\frac{\eta}{1+r_{0,1}+\eta}\cdot \rho\left(\frac{E[L_1]-L_1}{1+r_{0,1}}\right)
\end{equation}
This can be viewed as the sum of a ``best estimate" and a ``risk margin", where the risk margin is given by the discounted capital cost for the discounted deviation of $L_1$ from $E[L_1]$. However, a separation between best estimate and risk margin is not an essential element of the framework and can in fact make calculations more complicated. 

We refer to the introduction section of Engsner et al. \cite{Engsner2020} and the references provided there for some of the many papers on the valuation of (insurance) liabilities. As in \cite{Engsner2020}, the current paper considers successive one-year capital requirements and insolvency and, in common also with Albrecher et al. \cite{Albrecher2022}, a key element is the definition and valuation of capital investments.

In the current paper, in line with M\"ohr \cite{Moehr2011} and unlike several other approaches to the valuation of insurance liabilities, the value of the liabilities is considered as a whole. A breakdown into e.g. a best estimate and a risk margin is not an essential element. A key difference to \cite{Moehr2011}, as well as to Engsner et al. \cite{Engsner2020} and others including Albrecher et al. \cite{Albrecher2022}, is that the production strategies are not restricted by assuming that the assets corresponding to the capital are invested separately and specifically in a numeraire asset. The current paper is intentionally more elementary than many other papers, not assuming the existence of (``risk-neutral") pricing measures and numeraires. A specific difference to Engsner et al. \cite{Engsner2020} is the specification of the capital investment, which in our case is set up for every one-year period with payout equal to the positive part of the value of assets over liabilities. In \cite{Engsner2020}, it consists of dividend streams until the option to default is exercised by capital investors and, in our understanding, does not consider the liability cash flows after the date considered. 

The paper is organized as follows. In Section~\ref{sec-setup}, we set up the framework, specifically cash flows, tradables, illiquid assets, and investment strategies, with specific consideration of portfolios containing negative units of tradables. In Section~\ref{sec-setup}, we set up the framework, specifically cash flows, tradables, illiquid assets, and investment strategies, with specific consideration of portfolios containing negative units of tradables. We introduce a normative condition on the market prices of tradables that is roughly comparable to ``no arbitrage" and that we call ``consistency". Consistency essentially means that if, at a given point in time, the market price of a suitable strategy is not smaller than that of another suitable strategy (almost surely), then this holds at any earlier point in time.  

Section~\ref{sec-productionofliabilitieswithcapital} sets out the production of liabilities in the framework, introducing the valuation of assets and liabilities in the balance sheet, fulfillment and the financiability condition, and the central definition of a production strategy, at first limited to strategies with non-negative units of tradables and that stop whenever failure occurs. This is extended in Section~\ref{subsec-extensionofadmstrategies} to strategies that can contain negative units and in Section~\ref{ss-prodstratinclfailure} to strategies that continue after failure.   

In Section~\ref{sec-propertiesofvalue}, we analyze properties of the valuation, specifically, in Section~\ref{subsec-perfectreplication}, under which conditions the market prices of tradables are recovered. Section~\ref{ss-addingshortpos} looks at adding short positions to liabilities and Section~\ref{subsec-prodwithilliquidassets} investigates production also using illiquid assets. Section~\ref{ss-valuesolvencyiisst} sets out how the valuation of insurance liabilities under Solvency II and in the SST can be seen as special cases of the framework. Section~\ref{sec-conclusion} concludes.

\section{Set-up of the framework}\label{sec-setup}

\subsection{General}

Given a natural number $T$, we consider a finite discrete set $D[0,T]$ of \emph{dates} $t_j\in [0,T]$ ($j=0,\ldots ,J$) that includes the start dates $t=i$ and end dates $t=i+1$ of the calendar years $i\in\{0,1,2,\ldots, T-1\}$ and at least another date within each calendar year, 
\begin{eqnarray*}
D[0,T] & = & \{ t_0=0\leq i=t_{j}<t_{j+1}<i+1=t_{\ell} \leq t_{J}=T\mid \\
& & \mbox{for } i\in\{0,\ldots ,T-1\}\} 
\end{eqnarray*}
We define the subset of dates $D[t_{min},t_{max}]=D[0,T]\cap \interval{t_{min}}{t_{max}}$ for dates ${t_{min}}, {t_{max}}\in D[0,T]$ with $t_{min}< t_{max}$. For $t_j\in D[0,T]$ with $t_j<T$ we define the successor date $\gamma  (t_j)$ by $\gamma  (t_j)=t_{j+1}$, and $\gamma  (T)=T+1$. We assume a filtered probability space defined by a probability space $(\Omega, \mathcal{F}, \mathbb{P})$ and a filtration $\mathbb{F}=(\mathcal{F}_t)_{t\in D[0,T]}$ consisting of an increasing sequence of $\sigma$-fields $\mathcal{F}_t$,
\begin{equation*}
\{ \emptyset, \Omega\}=\mathcal{F}_0\subseteq \mathcal{F}_{t_j} \subseteq \mathcal{F}_{t_{j+1}}\subseteq\mathcal{F}_T=\mathcal{F}
\end{equation*}
The field $\mathcal{F}_t$ is assumed to correspond to the information available (known) at date $t$. We assume zero sets are measurable. Unless otherwise specified, equalities and inequalities are assumed to hold $\mathbb{P}$-almost surely. Stochastic processes are assumed to be \emph{adapted} to the filtration $\mathbb{F}$ (e.g. $X_t$ is $\mathcal{F}_s$-measurable for $s\geq t$) except for investment strategies $\phi=(\phi_t)_t$ (see below), which by convention are assumed to be \emph{predictable}, i.e. the strategy denoted $\phi_{t_{j+1}}$ from $t_j$ to $t_{j+1}$ is $\mathcal{F}_{t_j}$-measurable.

We denote by $\mathbb{R}^{d+1}_{\geq 0}$ for $d\in\{0,1,\ldots\}$ the set of $x=(x_0,\ldots ,x_d)\in\mathbb{R}^{d+1}$ with $x_k\geq 0$ for all $k\in\{0,1,\ldots , d\}$. For $a,b\in\mathbb{R}$, we write $(a )_{+}=\max\{0,a\}$ and $a\wedge b=\min\{ a,b\}$. For a subset $A\in\mathcal{F}$, the indicator function $1_A$ is $1_A(\omega)=1$ for $\omega\in A$ and $1_A(\omega)=0$ otherwise. We write ``iff" for ``if and only if".

{\bf Cash flows.} We distinguish non-negative cash inflows, generally denoted by the letter $Z_t$, typically from assets, from non-negative cash outflows $X_t$, typically from liabilities, at dates $t\in D[0,T]$. We refer to in- and outflows often in relation to the insurance company under consideration but also sometimes in relation to an investment strategy (Definition~\ref{def-strategywithcashflows}). We distinguish \emph{contractually specified} non-negative cash flows $Z_t$ and $X_t$ from non-negative \emph{actual} cash flows $\tilde{Z}_t$ and $\tilde{X}_t$ that consider failure. Cash inflows from selling and cash outflows from buying an asset are not considered as cash flows ``of'' the asset. We consider a cash flow simply as a ``flow'' of tradables (see below) with a specific value, not restricting to cash in a specific currency.

{\bf Tradables.} We consider a financial market with $d+1$ \emph{tradables} defined as financial asset instruments enumerated by $k=0,1,\dots, d$ with $\mathcal{F}_t$-measurable non-negative reliable market prices $S_t=(S_t^0,S_t^1,\ldots ,S_t^d)$ in $\mathbb{R}^{d+1}_{\geq 0}$ with $S_t\neq 0$ at dates $t\in D[0,T]$. We make in general no further assumptions, specifically about the tradable $k=0$. The market price of an instrument is \emph{reliable} (as an idealization) if the instrument is ``fully liquid'', i.e. any quantity can be bought or sold instantly for a cash flow equal to the market price (disregarding transaction costs and other market frictions). For the company holding them as assets, tradables may generate $\mathcal{F}_t$-measurable actual cash inflows $\tilde{Z}_t^k\geq 0$ ($k\in\{ 0,\dots, d\}$) such as coupons for bonds and dividends for equities, for simplicity only at dates $t\in D[0,T]$. As a further simplification, we assume that the tradables have no cash outflows. We use the convention that the market price $S_t^k$ at a cash flow date $t\in D[0,T]$ \emph{does not include}, i.e. is ``immediately after" the cash inflow $\tilde{Z}_t^k$. Because cash flows are relevant in the framework for default and for illiquid assets, we use as $S_t^k$ the actually observed market prices and not total return indices derived by assuming that cash inflows are automatically reinvested in the tradable. 

In view of the asymmetry between assets and liabilities, we first consider in Section~\ref{ss-nonnegstrat} only non-negative (investment) strategies of tradables (i.e. assets), that is, with only non-negative units of tradables. In Section~\ref{ss-genstratrestrilliquid}, we then define general strategies that allow for negative units of tradables.

\subsection{Non-negative strategies, illiquid assets}\label{ss-nonnegstrat}
 
{\bf Non-negative strategies.} Strategies for producing insurance liabilities in general include cash inflows $\tilde{Z}_t$ and outflows ${X}_t$ from the insurance liabilities such as premiums and claims payments, respectively, that do not in general cancel each other out. For this reason, we consider non-negative strategies $\phi$ of tradables with finite discrete trading dates $t\in D[t_{min},t_{max}]$ that are not necessarily self-financing. A non-negative (asset) portfolio $\phi_{t_{j+1}}\geq 0$  of tradables for $t_j\in D[t_{min},t_{max}]$ with $t_j<t_{max}$ is defined to be an $\mathcal{F}_{t_j}$-measurable random vector on $\mathbb{R}^{d+1}_{\geq 0}$. Abbreviating $t\equiv t_{j+1}$, in particular, the number $\phi_{t}^{k}$ of units of any tradable $k$ is non-negative. The (actual) cash inflows of the portfolio $\phi_{t}$ are $\tilde{Z}_{t}^{\phi_{t}}=\sum_{k=0}^{d} \phi_{t}^{k}\cdot \tilde{Z}_{t}^{k}$, and we write $\phi_{t}\cdot S_s=\sum_{k=0}^{d} \phi_{t}^{k}\cdot S_s^{k}$.
\begin{defn}[Non-negative investment strategies, self-financing, terminal value]\label{def-strategywithcashflows}
For $D\equiv D[t_{min},t_{max}]$, a \emph{non-negative (investment) strategy} $\phi=(\phi_t)_{t\in D}$ of tradables with cash inflows $(\tilde{Z}_t)_{t\in D}$ and outflows $({X}_t)_{t\in D}$ to the strategy is a predictable sequence of non-negative portfolios $\phi_{t_j}$ for ${t_j}\in D\cup \{\gamma (t_{max})\}$, i.e. $\phi_{t_{min}}$ is $\mathcal{F}_{t_{min}}$-measurable and $\phi_{t_{j}}$ is $\mathcal{F}_{t_{j-1}}$-measurable for $t_j>t_{min}$ in $D\cup \{\gamma (t_{max})\}$, such that at any date $t_{j}\in D$,
\begin{equation}\label{eqn-selffinancing}
	\phi_{t_{j+1}}\cdot S_{t_{j}} = \phi_{t_{j}}\cdot S_{t_{j}} + \tilde{Z}_{t_{j}}^{\phi} + \tilde{Z}_{t_{j}} - X_{t_{j}}
\end{equation}
where we define $\tilde{Z}_{t_{j}}^{\phi}=\tilde{Z}_{t_{j}}^{\phi_{t_{j}}}$. The value $v_{t_{j}}(\phi)$ of a strategy $\phi$ at $t_{j}\in D$ is defined to be the market price of the portfolio $\phi_{t_{j+1}}$ after the cash flows at date $t_{j}\in D$, 
\begin{equation}\label{def-valuestrategy}
		v_{t_{j}}(\phi)=\phi_{t_{j+1}}\cdot S_{t_{j}}
\end{equation}
The value $v_{t_{max}}(\phi)$ is called \emph{terminal value} of the strategy $\phi$ on $D$.

The strategy $\phi$ is \emph{self-financing} iff $\tilde{Z}_{t_{j}}=X_{t_{j}}$ for any $t_{j}\in D$, i.e. \eqref{eqn-selffinancing} becomes
\begin{equation}\label{eqn-strictlyselffinancing}
	\phi_{t_{j+1}}\cdot S_{t_{j}} = \phi_{t_{j}}\cdot S_{t_{j}} + \tilde{Z}_{t_{j}}^{\phi}
\end{equation}
\end{defn}
\begin{remark} Note that we do not consider the cash inflows $\tilde{Z}_{t_{j}}^{\phi}$ to be part of the cash inflows $(\tilde{Z}_t)_{t\in D}$ of the strategy. Also note that \eqref{eqn-selffinancing} in general does not determine the portfolio $\phi_{t_{j+1}}$. 
\end{remark}

As an illustration, consider units of only one tradable $k\in\{0,\ldots ,d\}$ with non-zero cash inflows $\tilde{Z}_{t}^{\phi^k}\neq 0$. The strategy to invest the cash inflows $\tilde{Z}_{t}^{\phi^k}$ in additional units of the tradable $k$ has cash flows $\tilde{Z}_t={X}_t=0$ incoming to and outgoing of the strategy, respectively, and is thus self-financing according to Definition~\ref{def-strategywithcashflows}. In contrast, the strategy of passively holding the units unchanged has $\tilde{Z}_t=0$ and ${X}_t=\tilde{Z}_{t}^{\phi^k}\neq 0$ and so is not self-financing. Here, the cash flows ${X}_t$ are outgoing relative to the strategy but may remain within the company. 

In our framework, the following property of ``consistency'' of a set of tradables takes a role similar to the classical no-arbitrage condition.
\begin{defn}[Consistent tradables]\label{def-consistenttradables}
A set of tradables is \emph{consistent} on $D\equiv D[t_{min},t_{max}]$ iff for any two strategies $\phi$ and $\theta$ on $D$ with tradables in this set and any dates $t, \gamma (t)\in D$,
\begin{equation}\label{eqn-consistencyconditionsimple}
	 \phi_{\gamma (t)}\cdot S_{\gamma (t)} +\tilde{Z}_{\gamma (t)}^{\phi}\geq \theta_{\gamma (t)}\cdot S_{\gamma (t)} +\tilde{Z}_{\gamma (t)}^{\theta} \Rightarrow \phi_{\gamma (t)}\cdot S_{t}\geq \theta_{\gamma (t)}\cdot S_{t}
\end{equation}
\end{defn}
This is a condition on tradables because the strategy is not changed between the two dates $t$ and $\gamma (t)$. 
The corresponding property for strategies is the following: if, at some date $t_j\in D$, the values of the two strategies $\phi$ and $\theta$ satisfy $v_{t_{j}}(\phi)\geq v_{t_{j}}(\theta)$, then $v_{t_{\ell}}(\phi)\geq v_{t_{\ell}}(\theta)$ at any earlier date $t_{\ell}\in D$ with $t_{\ell}<t_j$. This property depends on consistency of tradables as well as on the conversion of the strategies according to \eqref{eqn-selffinancing} at the intermediate dates, as we now show. To investigate when the property holds, consider two non-negative strategies $\phi$ and $\theta$ with cash flows $\tilde{Z}_t, X_t$ and $\tilde{Z}_t^{\prime}, X_t^{\prime}$, respectively, such that $v_{t_{j}}(\phi)\geq v_{t_{j}}(\theta)$. Using the definition \eqref{def-valuestrategy} of the value and \eqref{eqn-selffinancing}, this implies $\phi_{t_{j}}\cdot S_{t_{j}}+\tilde{Z}_{t_{j}}^{\phi}+\tilde{Z}_{t_{j}}- X_{t_{j}} \geq  \theta_{t_{j}}\cdot S_{t_{j}}+\tilde{Z}_{t_{j}}^{\theta}+\tilde{Z}_{t_{j}}^{\prime}-X_{t_{j}}^{\prime}$. If we then assume that $\phi$ and $\theta$ have the same ``net cash flows", i.e. 
\begin{equation}\label{eqn-samenetcashflows}
	\tilde{Z}_{t}-X_{t}=\tilde{Z}_{t}^{\prime}-X_{t}^{\prime}\;\mbox{ for any }t\in D\mbox{ with }t_{\ell}<t\leq t_j
\end{equation}
we can apply consistency of tradables \eqref{eqn-consistencyconditionsimple} to get $v_{t_{j-1}}(\phi)\geq v_{t_{j-1}}(\theta)$. Continuing the argument recursively, we get $v_{t_{\ell}}(\phi)\geq v_{t_{\ell}}(\theta)$. Note that \eqref{eqn-samenetcashflows} clearly holds for self-financing strategies.

We can show similarly: among strategies with consistent tradables in the sense of \eqref{eqn-consistencyconditionsimple}, the value is uniquely determined by the terminal value and the net cash flows. Explicitly, two strategies $\phi$ and $\theta$ on $D\equiv D[t_{min},t_{max}]$ with the same terminal value $v_{t_{max}}(\phi)=v_{t_{max}}(\theta)$ (e.g. zero) and the same net cash flows $\tilde{Z}_{t}-X_{t}=\tilde{Z}_{t}^{\prime}-X_{t}^{\prime}$ for any $t\in D$ have the same value $v_{t}(\phi)=v_{t}(\theta)$ for any $t\in D$. 

{\bf Illiquid assets.} In addition to tradables, we consider a non-negative portfolio $\psi=(\psi_0^1,\ldots ,\psi_0^{l})\in\mathbb{R}^{l}_{\geq 0}$ of a finite number of \emph{``fully" illiquid assets} at date $t=0$, making the extreme assumption that these cannot be sold or bought, so $\psi$ is ``held to maturity" (HTM). We assume that any illiquid asset $\ell$ has $\mathcal{F}_t$-measurable actual cash inflows $\tilde{Z}_{t}^{\ell}\geq 0$ and no cash outflows. The corresponding non-negative strategy also denoted by $\psi$ has cash flows $\tilde{Z}_{t}^{\psi}$ at dates $t\in D[0,T]$ given by $\tilde{Z}_{t}^{\psi}=\sum_{\ell=1}^{l}\psi_0^{\ell}\cdot \tilde{Z}_{t}^{\ell}$. This corresponds to a (not self-financing) investment strategy as in Definition~\ref{def-strategywithcashflows} with no cash inflows $\tilde{Z}_t=0$ and cash outflows $X_t=\tilde{Z}_{t}^{\psi}$, i.e.  with $\psi_{\gamma (t)}\cdot S_t=\psi_t\cdot S_t$ corresponding to \eqref{eqn-selffinancing}.

\subsection{Short positions, general strategies, restrictions}\label{ss-genstratrestrilliquid}

{\bf Short position.} For positive units $\phi^k>0$ of the tradable $k$ as an asset, the corresponding negative units $\theta^k=-\phi^k<0$ by ``going short" are a liability. The asset that is the ``counterposition" to this liability in general is different from the original asset. E.g. ``going short" in cash, i.e. borrowing cash by some debt contract produces an asset position of the cash lender in the debt contract that is different from simply holding cash. This illustrates that asset and liability from ``going short" need to be defined through an agreement between the counterparties, with per-se several possible definitions. (Further, the agreement is exposed to failure of the liability counterparty.) In the definition we choose, the original asset is in particular "replicated" in the following sense: the asset counterparty of the agreement receives the actual (not the contractually specified) cash inflows $\phi^k\cdot\tilde{Z}_t^k$ of the units of the tradable $k$ (e.g. bond coupons, dividends) and in addition has the option at any time to get the current market price $\phi^k\cdot S_t^k$ and any cash inflows $\phi^k\cdot\tilde{Z}_t^k$ of the position in cash, which terminates the agreement. To use negative units of tradables in production, we need to assume that such agreements can be set up. We may further assume that they can be closed out.

We define short positions directly as a liability $\mathcal{L}(\phi)$ for non-negative strategies $\phi$ as in Definition~\ref{def-strategywithcashflows} with cash inflows $\tilde{Z}_t=0$ and cash outflows $X_t$ to the strategy $\phi$ but that are assumed to initially remain with the company. We consider general cash outflows $X_t$, with the cash inflows $\tilde{Z}_t^{\phi}$ going out being a typical special case. We define $\mathcal{L}(\phi)$ as a ``replication" of the strategy $\phi$ to the corresponding asset holder, in particular with the cash flows $X_t$ paid out to the asset holder, with an option for the asset holder to terminate the agreement and get the value in cash. Short positions in a tradable are a special case.  
\begin{defn}[Short position in non-negative strategy, available for production]\label{def-negativeshares}
Let $\phi\geq 0$ be a non-negative strategy on $D\equiv D[t_{min},t_{max}]$ with cash outflows $X_t$ and no cash inflows $\tilde{Z}_t=0$. The short position liability $\mathcal{L}({\phi})$ on $D$ is defined using a stopping time $\tau\in D$ by cash inflows $\tilde{Z}_t^{\mathcal{L}(\phi)}=0$ ($t\in D$) and cash outflows for $t\in D$,
\begin{equation}\label{def-assetstrategycashflows}
	X_t^{\mathcal{L}(\phi)} =X_t\mbox{ for }\; t<\tau\mbox{; }\; \phi_{\tau}\cdot S_{\tau} +\tilde{Z}_{\tau}^{\phi}\mbox{ for }t=\tau\mbox{; } 0\mbox{ for }t>\tau
\end{equation}	
The liability $\mathcal{L}(\phi)$ is assumed to be extinguished after date $\tau$, so its terminal value at $\tau$ is zero.

Portfolios containing negative units of a tradable $k\in\{0,1,\ldots ,d\}$ are said to be \emph{available for production} on $D$ iff it is possible at any $t\in D$, for the market price of the corresponding positive units of the tradable $k$, to borrow the units of the tradable $k$ by taking on the corresponding liability, and \emph{available for production with close out} iff it is in addition possible at any $t\in D$ for the market price to close out (pay back) such an existing liability.  
\end{defn}
Clearly, the liability $\mathcal{L}(\phi)$ can be produced by holding the strategy $\phi$ until date $\tau$ and in the meantime paying out the cash flows $X_t$.

Alternatively to the option of the asset holder to terminate the contract, leading to the stopping time in Definition~\ref{def-negativeshares}, we could simply consider liabilities $\mathcal{L}(\phi)$ where the value of the strategy is paid at date $t_{max}$, i.e. with the cash flow in \eqref{def-assetstrategycashflows} replaced by, for $t\in D$, 
\begin{equation}
	X_t^{\mathcal{L}(\phi)} =X_t\mbox{ for }t<t_{max}\mbox{; } \phi_{t_{max}}\cdot S_{t_{max}} +\tilde{Z}_{t_{max}}^{\phi}\mbox{ for }t=t_{max}\mbox{; } 0\mbox{ for }t>t_{max}
\end{equation}	
If we assume that portfolios containing negative units of tradables are available for production with close out (Definition~\ref{def-negativeshares}), then it is sufficient to consider this alternative, as the more general case can be produced by simply closing out the strategy at date $t=\tau$, because by closing out, the market price is paid out, i.e. $\phi_{\tau}\cdot S_{\tau} +\tilde{Z}_{\tau}^{\phi}$. This holds if we additionally assume that the liability is extinguished after close out, i.e. its terminal value at $\tau$ is zero.

{\bf General strategy.} A non-negative strategy as in Definition~\ref{def-strategywithcashflows} is characterized by asset portfolios and the portfolio conversion equation \eqref{eqn-selffinancing}. For the extension to general strategies also containing short positions, we formally decompose portfolios $\phi_t\in\mathbb{R}^{d+1}$ into $\phi_t=\phi_t^{+}-\phi_t^{-}$ with non-negative portfolios $\phi_t^{\pm}=(\pm \phi_t)_{+}\geq 0$ and interpret this as an asset portfolio $\phi_t^{+}$ together with a liability $\mathcal{L}^{*}(\phi^{-})$ to be defined through its cash flows. Inserting $\phi_t=\phi_t^{+}-\phi_t^{-}$ into the formal portfolio conversion equation \eqref{eqn-selffinancing} for $\phi$ at date $t\in D$ and rearranging we get a portfolio conversion equation for the non-negative strategy $\phi^{+}$:
\begin{equation}\label{eqn-genstrat}
	\phi_{\gamma (t)}^{+}\cdot S_{t}=\phi_{t}^{+}\cdot S_{t} +\tilde{Z}^{\phi^{+}}_{t}+\tilde{Z}_{t}-X_{t} - (\phi_{t}^{-}\cdot S_{t}+\tilde{Z}^{\phi^{-}}_{t}) + 		(\phi_{\gamma (t)}^{-}\cdot S_{t}) 
\end{equation}
In view of this, we define the cash flows for the liability $\mathcal{L}^{*}\equiv\mathcal{L}^{*}(\phi^{-})$ at dates $t$ to be the cash outflow $X_{t}^{\mathcal{L}^*}=\phi_{t}^{-}\cdot S_{t}+\tilde{Z}^{\phi^{-}}_{t}$ above from closing out at date $t$ a short position in the portfolio $\phi_{t}^{-}$ and the cash inflows $\tilde{Z}_{t}^{\mathcal{L}^{*}}=\phi_{\gamma (t)}^{-}\cdot S_{t}$ above from taking on a new short position in the portfolio $\phi_{\gamma (t)}^{-}$ provided that $t<t_{max}$.  
\begin{defn}[General strategy]\label{def-liabclosedtakenon}
Assume that portfolios containing negative units of tradables are available for production with close out. A strategy $\phi=\phi^{+}-\phi^{-}$ on $D\equiv D[t_{min},t_{max}]$ with cash flows $\tilde{Z}_t$ incoming to and cash flows $X_t$ outgoing of the strategy is defined as the non-negative strategy $\phi^{+}$ with cash flows $\tilde{Z}_t$ and $X_t$ together with the liability $\mathcal{L}^{*}\equiv\mathcal{L}^{*}(\phi^{-})$ defined by the cash inflows and outflows for any $t\in D$,
\begin{eqnarray}
	& & X_{t}^{\mathcal{L}^*} =\phi_{t}^{-}\cdot S_{t}+\tilde{Z}_{t}^{\phi^{-}}\\
	& & \tilde{Z}_{t}^{\mathcal{L}^{*}}=\phi_{\gamma (t)}^{-}\cdot S_{t}\mbox{ for }\; t<t_{max},\;\;\tilde{Z}_{t_{max}}^{\mathcal{L}^*}=0\label{def-tildeztforphi}
\end{eqnarray}	
The value of the strategy $\phi$ for $t\in D$, $t<t_{max}$ is defined as: $v_{t}(\phi)=v_{t}(\phi^{+})-v_{t}(\mathcal{L}^{*})$ with $v_{t}(\mathcal{L}^{*})=v_{t}(\phi^{-})$. 
\end{defn}
Consistency can be extended from non-negative to general strategies using consistency for non-negative strategies (Definition~\ref{def-consistenttradables}) and the same condition \eqref{eqn-samenetcashflows} on the cash flows. To show this, let $t\in D$ and assume $v_{\gamma (t)}(\phi)\geq v_{\gamma (t)}(\theta)$ for $\phi, \theta\in\mathcal{R}^{\prime}$. By the definition of the value from Definition~\ref{def-liabclosedtakenon} and \eqref{eqn-genstrat}, this is the same as $\phi_{\gamma (t)}^{+}\cdot S_{\gamma (t)} +\tilde{Z}^{\phi^{+}}_{\gamma (t)}+\tilde{Z}_{\gamma (t)}-X_{\gamma (t)} - (\phi_{\gamma (t)}^{-}\cdot S_{\gamma (t)}+\tilde{Z}^{\phi^{-}}_{\gamma (t)})\geq \theta_{\gamma (t)}^{+}\cdot S_{\gamma (t)} +\tilde{Z}^{\theta^{+}}_{\gamma (t)}+\tilde{Z}_{\gamma (t)}-X_{\gamma (t)} - (\theta_{\gamma (t)}^{-}\cdot S_{\gamma (t)}+\tilde{Z}^{\theta^{-}}_{\gamma (t)})$. Using \eqref{eqn-samenetcashflows}, we rearrange this inequality to get an inequality between the non-negative strategies $\phi^{+}+\theta^{-}$ and $\phi^{-}+\theta^{+}$ as on the right-hand side of \eqref{eqn-consistencyconditionsimple}, so consistency \eqref{eqn-consistencyconditionsimple} implies that $v_{t}(\phi^{+}+\theta^{-})\geq v_{t}(\phi^{-}+\theta^{+})$, so $v_{t}(\phi)\geq v_{t}(\theta)$.  

{\bf Restrictions.} In certain applications, strategies may have to be restricted e.g. to investing only in a subset of tradables, i.e. $\phi_t^k=0$ for one or several $k\in\{0,1,\ldots d\}$. 
\begin{defn}[Restrictions]\label{def-restrictions}
For a linear subspace $\mathcal{R}$ of $\mathbb{R}^{d+1}$ and a strategy $\phi$ on $D\equiv D[t_{min},t_{max}]$, we define, by slight abuse of notation, $\phi\in\mathcal{R}$ iff $\phi_t\in\mathcal{R}$ for any $t\in D$, $\phi\in\mathcal{R}^{\geq 0}$ iff $0\leq \phi_t\in\mathcal{R}$ for any $t\in D$, and $\phi\in \mathcal{R}^{\prime}$ iff $\phi_t\in\mathcal{R}$ and $v_t(\phi)\geq 0$ for any $t\in D$. (Note: $\phi\in\mathcal{R}\Rightarrow\phi^{\pm }\in\mathcal{R}$.)
\end{defn}

\section{Production of liabilities}\label{sec-productionofliabilitieswithcapital}

\subsection{Non-negative production strategies}\label{ss-nonnegprodstrategydefn}

{\bf Set-up.} We consider an insurance company that is subject to regulatory solvency capital requirements and insolvency laws. Roughly, the insolvency laws determine the circumstances under which the company fails and the consequences of failure, and the regulatory solvency capital requirements intend to make failure sufficiently unlikely. We link this to sufficient fulfillment of the insurance liabilities by production strategies: at specified dates, no failure is intended to be equivalent to sufficient fulfillment according to the fulfillment condition (Definition~\ref{def-fulfillmentcondition} below) that is encoded in the regulatory solvency capital requirements. We assume in line with typical rules that a company fails when it defaults on a payment that is due or is balance sheet insolvent, i.e. the value of its assets is less than the value of its liabilities. We additionally assume a financiability condition (Definition~\ref{def-financiabilitycondition} below) that specifies when the return from a capital investment is sufficient for capital investors. 

{\bf Insurance contracts.} Let $\mathcal{L}$ denote the insurance contracts that the company has written until date $0$, which we also sometimes imprecisely identify with the insurance liabilities from these contracts. The insurance contracts are assumed to be extinguished at date $T$, with no remaining cash flows or value, and we assume that the company writes no additional insurance contracts in the period $[0,T]$. We represent $\mathcal{L}$ by contractually specified cash outflows $X_t^{\mathcal{L}}\geq 0$ for claims and benefits and other e.g. administrative costs, and actual cash inflows $\tilde{Z}_t^{\mathcal{L}}\geq 0$ for premiums, at dates $t\in D[0,T]$. Our objective is to value the insurance liabilities $\mathcal{L}$ as a whole by the cost $\Bar{v}_{0}^{\phi,\psi,\mathcal{C}}(\mathcal{L})$ of production strategies $\phi$ that sufficiently produce them according to the fulfillment condition, using capital $\mathcal{C}=(C_i)_{i=0}^{T-1}$ and in general also illiquid assets $\psi$. As a simplification, we do not explicitly model the exercise of options by policyholders to terminate the insurance contracts and instead assume that this is implicitly taken into account in the cash flows $X_t^{\mathcal{L}}$ and $\tilde{Z}_t^{\mathcal{L}}$. 

{\bf Production of insurance contracts.} In the framework, insurance liabilities are fulfilled by production strategies that provide sufficient fulfillment. (Non-negative) production strategies as defined in the central Definition~\ref{def-admissibleinvestmentstrategy} below are more complicated than tradable investment strategies according to Definition~\ref{def-strategywithcashflows} with cash inflows $\tilde{Z}_t^{\mathcal{L}}$ and outflows $X_t^{\mathcal{L}}$ in particular due to sufficient fulfillment, capital, and fully illiquid assets. In the present section, we only consider strategies that contain only non-negative units of tradables and that do not cover what happens in case of failure. We additionally distinguish between strategies with non-negative and negative production cost, for reasons that we explain below. 

In line with the one-year view of regulatory capital requirements such as Solvency II and the SST, the production proceeds over successive one-year periods from $i-1$ to $i$ for $i=1,2,\ldots ,T$, with recapitalization assumed to be possible at each annual date $i-1$. For simplicity, we only allow for production strategies under which the insurance company (almost surely) does not default on the payments $X_{t}^{\mathcal{L}}$ at dates $t\in D[0,T]\setminus\{0,1,\ldots ,T\}$ and assume that balance sheet insolvency is only assessed at dates $i\in\{1,2,\ldots ,T\}$. Hence, by assumption, failure is only possible at the latter dates. 

At a date $i-1$, given no failure up to and including this date, the insurance company in general has a (non-negative) tradable asset portfolio $\phi_{\gamma (i-1)}$, an illiquid asset portfolio $\psi$, insurance liabilities $\mathcal{L}$, and capital arrangements with capital $C_{i-1}\geq 0$. From date $i-1$ to $i$, the company follows the given production strategy, which we denote by $\phi$. For this strategy, seen from date $i-1$, i.e. conditional on $\mathcal{F}_{i-1}$, the fulfillment condition is satisfied at $i$ and the financiability condition holds for the capital arrangement. Following the strategy, at each successive date $t\in D[0,T]$ with $i-1<t<i$, we assume the following sequence of steps: (i) cash inflows $\tilde{Z}_{t}^{\phi+\psi+\mathcal{L}}$ from the tradables $\phi$, the illiquid assets $\psi$, and the insurance contracts $\mathcal{L}$; (ii) the insurance obligations $X_{t}^{\mathcal{L}}$ are paid by selling tradables in $\phi_t$ or with the cash inflows $\tilde{Z}_{t}^{\phi+\psi+\mathcal{L}}$; and (iii) a new non-negative tradable portfolio $\phi_{\gamma (t)}$ is set up by potentially further trading tradables such that the following equality in terms of values holds:\footnote{Recall the illiquid assets $\psi$ can by assumption not be traded.} 
\begin{equation}\label{eqn-conversionredundant}
	\phi_{\gamma (t)}\cdot S_{t} = \phi_{t}\cdot S_{t} + \tilde{Z}_{t}^{\phi+\psi+\mathcal{L}} - X_{t}^{\mathcal{L}}\geq 0
\end{equation}
Non-negativity (almost surely) in \eqref{eqn-conversionredundant} corresponds to the above assumption that there is no default on the payments $X_{t}^{\mathcal{L}}$ for $i-1<t<i$.

At date $i$, the following additional aspects need to be taken into account: the company may fail; it may not be possible to continue the production strategy; and the value $C_{i}^{\prime}$ of the capital investment is assessed and potentially additional capital is raised. We make the following simplifying assumptions (which may not hold in reality due to different valuation standards): failure is defined by $A^{\prime}_{i}<L_{i}$ with the values $A^{\prime}_{i}$ of the assets and $L_{i}$ of the insurance liabilities according to Definition~\ref{def-assetsandliabilities} below, and the same valuation determines the value $C_{i}^{\prime}$ of the capital investment. The value of the capital investment is then given by $C_{i}^{\prime}=(A^{\prime}_{i}-L_{i})_{+}$ because the claims of capital investors are subordinated to all other obligations and non-negative. Non-negativity is a consequence of the limited liability of the capital investors, i.e. in case of failure $A^{\prime}_{i}<L_{i}$, the capital investors get nothing but are also not obliged to inject additional capital. Note that, under the assumptions, the capital investment cannot cause failure.

We assume the following sequence of steps at date $i$ and then explain why specific steps are possible: 
\begin{enumerate}[(1)]
\item  cash inflows $\tilde{Z}_{i}^{\phi+\psi+\mathcal{L}}$; 
\item failure is assessed by comparing $A^{\prime}_{i}$ to $L_{i}$. In case of no failure $A^{\prime}_{i}\geq L_{i}$, 
\item the insurance liability cash outflows $X_{i}^{\mathcal{L}}$ are provided;  
\item the compensation $C_{i}^{\prime}$ is provided to the capital investors; 
\item capital $C_{i}$ in the form of tradables is raised;\footnote{Steps (4) and (5) are presented as separate steps for convenience.} and 
\item the new tradable portfolio $\phi_{\gamma (i)}$ is derived by potentially further trading tradables, satisfying: 
\begin{equation}\label{eqn-conversionredundantip1}
	\phi_{\gamma (i)}\cdot S_{i} = \phi_{i}\cdot S_{i}+\tilde{Z}_{i}^{\phi+\psi+\mathcal{L}} - X_{i}^{\mathcal{L}} - C^{\prime}_{i} + C_{i}
\end{equation}
\end{enumerate}
The portfolio $\phi_{\gamma (i)}$ is the starting point for the production strategy for the next one-year period $i$ to $i+1$, which needs to be such that, conditional on $\mathcal{F}_{i}$, the fulfillment condition is satisfied at $i+1$, and the financiability condition is satisfied for the investment $C_{i}\rightarrow C_{i+1}^{\prime}$. The latter is why the capital $C_{i}$ can be raised in step 5, provided that moving to $\phi_{\gamma (i)}$ in step 6 is then possible. We explain below (Lemma~\ref{lemma-failure} following the definition of the values $A^{\prime}_{i}$ and $L_{i}$ in Definition~\ref{def-assetsandliabilities}) why it is possible to move from step 2 to steps 3 and 4 and then to step 6. The two equations \eqref{eqn-conversionredundant} and \eqref{eqn-conversionredundantip1} also show that, in case of no failure, production strategies are investment strategies $\phi$ in the sense of Definition~\ref{def-strategywithcashflows}, with cash inflows $\tilde{Z}_{i}^{\psi+\mathcal{L}}+C_{i}$ and outflows $X_{i}^{\mathcal{L}}+C^{\prime}_{i}$ at date $i$.

As we explain below, the above sequence of steps only works for strategies with non-negative production cost $\Bar{v}_{i}^{\phi, \psi,\mathcal{C}}(\mathcal{L})\geq 0$. To include negative production cost, which may e.g. occur when the outstanding premium cash inflows are relatively large, we introduce the following two additional steps 2bis between steps 2 and 3, and 5bis between 5 and 6: 
\begin{enumerate}
\item[(2bis)] tradables with value equal to $(-\Bar{v}_{i}^{\phi, \psi,\mathcal{C}}(\mathcal{L}))_{+}$ are borrowed by taking on a corresponding liability.
\item[(5bis)] the liability is paid back by tradables with value $(-\Bar{v}_{i}^{\phi, \psi,\mathcal{C}}(\mathcal{L}))_{+}$ from the capital $C_{i}$. 
\end{enumerate}
This obviously works if it is possible to take on and close out additional liabilities in this way. Without this assumption, we are restricted to production strategies with non-negative production cost, as we discuss in more detail in Section~\ref{ss-negprodcost}. In the following, we thus distinguish two situations: 
\begin{enumerate}[(A)]
\item the above assumption is made and production cost can also be negative and steps 2bis and 5bis are included; or 
\item the production strategies are restricted to non-negative production cost and steps 2bis and 5bis are not included.
\end{enumerate}
{\bf Production cost.} We define non-negative production cost $\Bar{v}_{i}^{\phi, \psi,\mathcal{C}}(\mathcal{L})\geq 0$ at date $i$ as the (minimum) value of non-negative asset portfolios of tradables at $i$ such that, if such a portfolio is available prior to step 5 above, then production with the given strategy $\phi$ is possible from date $i$. In the above sequence of steps, one would then raise capital $C_i$ in step 5 and convert the tradables in step 6 to the tradable portfolio $\phi_{\gamma (i)}$. It follows that the production cost is the (known) difference between the value $v_{i}(\phi)=\phi_{\gamma (i)}\cdot S_{i}$ and the capital $C_{i}$, i.e.  $\Bar{v}_{i}^{\phi, \psi,\mathcal{C}}(\mathcal{L})=\phi_{\gamma (i)}\cdot S_{i}-C_i$, provided that this difference is non-negative.  
 
Note that, although the production uses the illiquid assets $\psi$ and they in general reduce production cost in terms of the cost of the tradables required (Section~\ref{subsec-prodwithilliquidassets}), they are not part of the production cost $\Bar{v}_{i}^{\phi, \psi, \mathcal{C}}(\mathcal{L})$, which is limited to the value of the required tradables $\phi$. 

For negative production cost, the corresponding starting point at date $i$ is a tradable portfolio with negative value, i.e. to start with a liability and no tradable assets. Negative production cost $\Bar{v}_{i}^{\phi, \psi,\mathcal{C}}(\mathcal{L})<0$ at date $i$ can thus be defined as the negative of the (maximum) value of non-negative tradable portfolios that, in the additional step 5bis above, can be paid out immediately after raising the capital $C_i$ in step 5 above to settle the corresponding liability and then convert the remaining tradables with value $C_i-(-\Bar{v}_{i}^{\phi, \psi,\mathcal{C}}(\mathcal{L}))_{+}$ in step 6 to the tradable portfolio $\phi_{\gamma (i)}$. It again follows that the production cost is given by $\Bar{v}_{i}^{\phi, \psi,\mathcal{C}}(\mathcal{L})=\phi_{\gamma (i)}\cdot S_{i}-C_i$. Step 2bis and the need for steps 2bis and 5bis become clearer in the proof of Proposition~\ref{lemma-failure} below. 

{\bf Balance sheet valuation.} We define the values $A^{\prime}_{i}$ and $L_{i}$ in the following by considering the situation between steps 1 and 2 above, i.e. prior to the cash outflows $X_{i}^{\mathcal{L}}$ at date $i$; this allows accounting for failure caused by default. We then discuss the definition and show in Proposition~\ref{lemma-failure} below that no failure as defined by $A^{\prime}_{i}\geq L_{i}$ is equivalent to no default on $X_{i}^{\mathcal{L}}$ and the ability to continue the production strategy past date $i$. It follows from this in particular that steps 3 and 6 above are possible. 
\begin{defn}[Assets, liabilities]\label{def-assetsandliabilities}
Let $i\in\{1,2,\ldots,T\}$ and $\phi$ be a non-negative production strategy (Definition~\ref{def-admissibleinvestmentstrategy}) for the liabilities $\mathcal{L}$ using capital $\mathcal{C}$ and illiquid assets $\psi$ with production cost $\Bar{v}_{i}^{\phi, \psi,\mathcal{C}}(\mathcal{L})$. The value $A^{\prime}_{i}\geq 0$ of the assets and the value $L_{i}\geq 0$ of the liabilities are defined by (with $\tilde{Z}_{t}^{\phi+\theta}\equiv\tilde{Z}_{t}^{\phi}+\tilde{Z}_{t}^{\theta}$):
\begin{eqnarray}
    A^{\prime}_i = A^{\prime}_{i}(\mathcal{L},\phi,\psi,\mathcal{C}) & = & \phi_{i}\cdot S_{i} + \tilde{Z}_{i}^{\phi+\psi+\mathcal{L}}+(-\Bar{v}_{i}^{\phi, \psi,\mathcal{C}}(\mathcal{L}))_{+}\label{def-assetsi+1}\\
    L_{i} = L_{i}(\mathcal{L},\phi,\psi,\mathcal{C}) & = & X_{i}^{\mathcal{L}} + (\Bar{v}_{i}^{\phi, \psi,\mathcal{C}}(\mathcal{L}))_{+}\label{eqn-liabilitiesati+1}
\end{eqnarray}		
\end{defn}
\begin{remark} Definition~\ref{def-assetsandliabilities} contains all assets and liabilities at date $i$ in step 2 above explicitly or implicitly. The insurance contracts $\mathcal{L}$ are valued by their production cost as a liability when production cost are non-negative and as an asset otherwise. The illiquid assets $\psi$ and the cash inflows from $\mathcal{L}$, both of which by assumption can neither be bought nor sold, are accounted explicitly with their current cash inflows $\tilde{Z}_{i}^{\psi+\mathcal{L}}$ and implicitly in the production cost $\Bar{v}_{i}^{\phi, \psi,\mathcal{C}}(\mathcal{L})$ to the extent that the future cash inflows from $\psi$ and $\mathcal{L}$ are used to produce the future insurance liability cash outflows and thus reduce the production cost (see Section~\ref{subsec-prodwithilliquidassets}). 
\end{remark}

We now show how failure links to default and continuing the production strategy:
\begin{prop}\label{lemma-failure}
Failure at date $i\in\{1,2,\ldots T\}$ as defined by $A^{\prime}_{i}<L_{i}$ using Definition~\ref{def-assetsandliabilities} is equivalent to default on the payment $X_{i}^{\mathcal{L}}$ or the inability to continue the production strategy past date $i$.  
\end{prop}
\begin{proof}
We have $L_{i}\geq X_{i}^{\mathcal{L}}$ by definition. First assume no failure, $A^{\prime}_i\geq L_{i}$. In case of non-negative production cost, the value of the total tradable assets is $A^{\prime}_i=\phi_{i}\cdot S_{i} + \tilde{Z}_{i}^{\phi+\psi+\mathcal{L}}$, so they are sufficient to pay the cash outflows $X_{i}^{\mathcal{L}}$. In case of negative production cost, however, the assets with value $(-\Bar{v}_{i}^{\phi, \psi,\mathcal{C}}(\mathcal{L}))_{+}>0$ are not tradables, and given only $A^{\prime}_i\geq L_{i}$, the value $\phi_{i}\cdot S_{i} + \tilde{Z}_{i}^{\phi+\psi+\mathcal{L}}$ of the tradable assets is not necessarily sufficiently large to pay $X_{i}^{\mathcal{L}}$. This is where step 2bis above is used and links to the above definition of negative production cost: with the additional tradables with value equal to $-\Bar{v}_{i}^{\phi, \psi,\mathcal{C}}(\mathcal{L})> 0$ borrowed in step 2bis, the total value of the tradables is $A^{\prime}_i$, so the cash outflows $X_{i}^{\mathcal{L}}$ can be paid and there is no default. To show that the production strategy can be continued, note that the value $C_{i}^{\prime}$ of the capital investment considers all assets and liabilities potentially including the borrowed tradables and the corresponding liability but these cancel each other out and so $C_{i}^{\prime}=A^{\prime}_{i}-L_{i}$. Inserting the definitions of $A^{\prime}_{i}$ and $L_{i}$, rearranging terms to isolate $\Bar{v}_{i}^{\phi, \psi,\mathcal{C}}(\mathcal{L})$ on the right-hand side, and adding $C_i$ to both sides, we get $\phi_{i}\cdot S_{i} + \tilde{Z}_{i}^{\phi+\psi+\mathcal{L}}-X_{i}^{\mathcal{L}}-C_{i}^{\prime}+C_i= \Bar{v}_{i}^{\phi, \psi,\mathcal{C}}(\mathcal{L})+C_i=\phi_{\gamma (i)}\cdot S_{i}$ as required. 

Finally, assume failure $A^{\prime}_i< L_{i}$, we similarly get $\phi_{i}\cdot S_{i} + \tilde{Z}_{i}^{\phi+\psi+\mathcal{L}}< X_{i}^{\mathcal{L}}+ \Bar{v}_{i}^{\phi, \psi,\mathcal{C}}(\mathcal{L})$ and $C_{i}^{\prime}=0$. Possibly, there is then default on $X_{i}^{\mathcal{L}}$, but in any case, we get from the two previous relationships that $\phi_{i}\cdot S_{i} + \tilde{Z}_{i}^{\phi+\psi+\mathcal{L}}-X_{i}^{\mathcal{L}}-C_{i}^{\prime}+C_i< \Bar{v}_{i}^{\phi, \psi,\mathcal{C}}(\mathcal{L})+C_i=\phi_{\gamma (i)}\cdot S_{i}$, so the production strategy cannot be continued past date $i$. 
\end{proof}
\begin{remark} In our framework, failure at a date $i$ means default at this date or not sufficient fulfillment after the date, i.e. default can occur in the future ``more often" than allowed by the fulfillment condition. Proposition~\ref{lemma-failure} does not directly show this, because failure $A^{\prime}_i< L_{i}$ is defined for a specific production strategy $\phi$ and there may be other production strategies for the production after date $i$ with lower production cost. This can be overcome by changing the production strategy after date $i$ to a strategy with minimal production cost. Because finding such a minimum can be challenging, we allow in the following also for production strategies that are not minimal in this sense and may thus overestimate production cost or provide more capital compensation than would be needed. In line with this, fulfillment and financiability condition are defined below as minimum requirements.  
\end{remark}

{\bf Definitions.} We now define fulfillment and financiability condition and non-negative production strategies. Because of Proposition~\ref{lemma-failure}, sufficient fulfillment is defined in Definition~\ref{def-fulfillmentcondition} as a condition on the set on which $A^{\prime}_{i}\geq L_{i}$. Financiability conditions (Definition~\ref{def-financiabilitycondition}) express when the capital investment $C_{i-1}\rightarrow C_{i}^{\prime}=(A^{\prime}_{i}-L_{i})_{+}$ is acceptable to the capital investors. From now on, for convenience, we write a one-year period from date $i$ to $i+1$ instead of from $i-1$ to $i$.
\begin{defn}[Fulfillment condition]\label{def-fulfillmentcondition}
A {fulfillment condition} is a condition on the set $M_{i+1}=\{A^{\prime}_{i+1}\geq L_{i+1}\}\in\mathcal{F}_{i+1}$ ($i\in\{0,1,\ldots ,T-1\}$) \emph{conditional on} $\mathcal{F}_{i}$ with the property that if the set $M_{i+1}$ satisfies the condition, then so does any larger set $M\supseteq M_{i+1}$ in $\mathcal{F}_{i+1}$. The \emph{full fulfillment condition} corresponding to \emph{full production} is $\mathbb{P}[M_{i+1}\mid\mathcal{F}_{i}]=1$. 
\end{defn} 
For example (see Section~\ref{ss-valuesolvencyiisst}), the fulfillment condition under Solvency II is $\mathbb{P}[M_{i+1}\mid\mathcal{F}_{i}]\geq 0.995$. 

\begin{defn}[Financiability condition, consistent with and neutral to the tradables]\label{def-financiabilitycondition}
A financiability condition is a condition, for $i\in\{ 0,1,\ldots ,T-1\}$, on when the stochastic return from the capital investment
\begin{equation}\label{expr-payoffcapital}
     0\leq C_{i}\rightarrow C_{i+1}^{\prime}=(A^{\prime}_{i+1}-L_{i+1})_{+}  
\end{equation}
is sufficient for the capital investors, that has the following properties:
\begin{enumerate}[(a)]
\item The investment $0\rightarrow 0$ satisfies the financiability condition.
\item If $C_{i}\rightarrow C_{i+1}^{\prime}$ satisfies the financiability condition, then so does $C_{i}^{*}\rightarrow C_{i+1}^{*\prime}\,$ for any $0\leq C_{i}^{*}\leq C_{i}$ and $C_{i+1}^{*\prime}\geq C_{i+1}^{\prime}$. 
\end{enumerate}
We define:  
\begin{enumerate}[(1)]  
\item A financiability condition is \emph{positively homogeneous} iff, whenever $C_i\rightarrow C_{i+1}^{\prime}$ satisfies the financiability condition, then so does $\lambda_i\cdot C_i\rightarrow \lambda_i\cdot C_{i+1}^{\prime}$ for $\mathcal{F}_i$-measurable $\lambda_i\geq 0$.
\item It is \emph{consistent with the tradables} iff, whenever $C_{i}\rightarrow C_{i+1}^{\prime}$ satisfies the financiability condition, the strategy $(\phi_t)_{t=i}^{i+1}\in\mathcal{R}^{\prime}$ is self-financing, and $C_{i+1}^{\prime}\leq \phi_{i+1}\cdot S_{i+1}+\tilde{Z}_{i+1}^{\phi}$, then $C_{i}\leq v_i(\phi)$.
\item It is \emph{neutral to the tradables} iff, whenever $C_{i}\rightarrow C_{i+1}^{\prime}$ satisfies the financiability condition and the strategy $(\phi_t)_{t=i}^{i+1}\in\mathcal{R}^{\prime}$ is self-financing, then so does $C_i+v_i(\phi)\rightarrow C_{i+1}^{\prime}+\phi_{i+1}\cdot S_{i+1}+\tilde{Z}_{i+1}^{\phi}$.
\end{enumerate}  
\end{defn}
Definition~\ref{def-financiabilitycondition} also extends consistency for tradables from Definition~\ref{def-consistenttradables} to include capital investments. This is in particular used in Section~\ref{subsec-perfectreplication}. Intuitively, consistency with the tradables means that a capital investor would not invest more than the market price for the payoff of a self-financing strategy of tradables. Neutrality to the tradables, which is in particular used in Section~\ref{subsec-prodwithilliquidassets}, intuitively means that a capital investor would accept to additionally invest in the payoff of a self-financing strategy of tradables for their market price. 

We now define non-negative production strategies.
\begin{defn}[Non-negative production strategy]\label{def-admissibleinvestmentstrategy}
Let the following be given: dates $i_{min}<i_{max}$ in $\{0,1,\ldots ,T\}$, insurance contracts $\mathcal{L}$ with $\mathcal{F}_{i_{max}}$-measurable terminal value $Y_{i_{max}}$, a fulfillment and a financiability condition, and a set $M\in \mathcal{F}_{i_{min}}$. A \emph{non-negative production strategy} $\phi\equiv \phi (\mathcal{C}, \psi, D, M)$ for producing the liabilities $\mathcal{L}$ on $D\equiv D[i_{min},i_{max}]$ and $M$ using capital $\mathcal{C}$ and illiquid assets $\psi$ is a non-negative strategy $\phi=(\phi_t)_{t\in D}\in\mathcal{R}^{\geq 0}$ together with capital amounts $\mathcal{C}=(C_i)_{i=i_{min}}^{i_{max}-1}$ with $C_i=C_i(\phi,\psi)\geq 0$ and illiquid assets $\psi$ such that, on $M$, for any $i\in\{i_{min},\ldots ,i_{max}-1\}$, whenever $A^{\prime}_{i}\geq L_{i}$ according to Definition~\ref{def-assetsandliabilities}, the following conditions are satisfied:
\begin{enumerate}[(a)]  
	\item\label{cond-genfinancing} For any $t\in D$ with $i< t< i+1$,  
	\begin{equation}\label{def-generalfinancing}
		\phi_{\gamma (t)}\cdot S_{t} = \phi_{t}\cdot S_{t} + \tilde{Z}_{t}^{\phi+\psi+\mathcal{L}} - X_{t}^{\mathcal{L}}\geq 0
	\end{equation}
	\item\label{FinCati} The capital investment $C_{i}\rightarrow C_{i+1}^{\prime}$ satisfies the financiability condition. 
	\item\label{FulCati} At $i+1$, the fulfillment condition is satisfied. 
\end{enumerate}	
The \emph{production cost} of the liabilities $\mathcal{L}$ at date $i\in\{i_{min}, \ldots,i_{max}-1\}$ is defined with $v_i(\phi)=\phi_{\gamma  (i)}\cdot S_i$ as:
\begin{equation}\label{def-prodcost}
	\Bar{v}_{i}^{\phi, \psi, \mathcal{C}}(\mathcal{L})=v_i(\phi)-C_i
\end{equation}
If there is no tradable for which portfolios containing negative units of the tradable are available for production with close out (Lemma~\ref{def-negativeshares}), production strategies must have non-negative production cost $\Bar{v}_{i}^{\phi, \psi, \mathcal{C}}(\mathcal{L})\geq 0$. 
\end{defn}
\begin{remark} In case production strategies in Definition~\ref{def-admissibleinvestmentstrategy} must have non-negative production cost, a production strategy with zero production cost can be derived from a production strategy with negative production cost by decreasing the capital investment to $C_i-(-\Bar{v}_{i}^{\phi, \psi, \mathcal{C}}(\mathcal{L}))=v_i(\phi)$, which by assumption is non-negative and less than $C_i$. 
\end{remark}
\begin{remark} The terminal value $Y_{i_{max}}$ is typically zero if there are no non-zero cash flows of $\mathcal{L}$ after date $i_{max}$, specifically after terminal date $T$, and can then be disregarded. Non-zero $Y_{i_{max}}$ (and the set $M$) are in particular needed for patching together production strategies (Section~\ref{ss-constrprodstrat}). 
\end{remark}

{\bf Value of insurance contracts.} The value of insurance contracts $\mathcal{L}$ is defined as the minimum or more generally the essential infimum over the applicable production strategies. We do not study in this paper when the essential infimum is given by a minimal production strategy.  
\begin{defn}[Value of insurance contracts]\label{def-valuebar}
Given fulfillment and financiability condition, illiquid assets $\psi$, and restrictions $\mathcal{R}^{\geq 0}$, the value $\Bar{v}_{i}^{\psi}(\mathcal{L})$ at $i\in\{0,\ldots ,T-1\}$ of insurance contracts $\mathcal{L}$ with terminal value $Y_{T}$ is defined as
   \begin{equation}
        \Bar{v}_{i}^{\psi}(\mathcal{L}) = \operatorname*{inf~ess}\{\Bar{v}_{i}^{\phi, \psi,\mathcal{C}}(\mathcal{L})\}
   \end{equation}
with the essential infimum taken over all production strategies $\phi$ from $i$ to $T$ with illiquid assets $\psi$ for the given fulfillment and financiability condition and restrictions $\mathcal{R}^{\geq 0}$. By convention, $\operatorname*{inf~ess} \emptyset :=+\infty$.
\end{defn}

\subsection{Negative production cost and future cash inflows}\label{ss-negprodcost}

Negative production cost are related to the question whether future cash inflows, such as premiums or inflows from illiquid assets, can be used to pay for insurance liability cash outflows occuring at earlier dates. This is a problem that may not always be captured by solvency capital requirements. As a very simple illustrative example, we consider two one-year periods from dates $i=0$ to $i=1$, and $i=1$ to $i= 2$, and deterministic insurance contract premium cash inflows $(\tilde{Z}_{1}^{\mathcal{L}},\tilde{Z}_{2}^{\mathcal{L}})=(0,100)$ and claims cash outflows $(X_{1}^{\mathcal{L}},X_{2}^{\mathcal{L}})=(10,0)$. Although the premium inflows $\tilde{Z}_{2}^{\mathcal{L}}$ significantly exceed the claims outflows $X_{1}^{\mathcal{L}}$, they can obviously not be used directly to pay the claims at date $i=1$ because they occur at the later date $i=2$. So, without additional tradables, default occurs at $i=1$. In line with this, without further assumptions, the production cost at $i=0$ is positive (e.g. equal to 10 discounted risk-free from $i=1$ to $i=0$). At the same time, the production cost at $t=1$ (which are after the cash flow $X_{1}^{\mathcal{L}}=10$) would be negative, as the cash inflows at $t=2$ of $\tilde{Z}_{2}^{\mathcal{L}}=100$ exceed the outflows of $X_{2}^{\mathcal{L}}=0$. The assumption about taking on and closing out additional liabilities set out in Section~\ref{ss-nonnegprodstrategydefn} allows ``pulling back" the negative production cost at $t=1$ to $t=0$ by allowing to borrow tradables at $t=1$ with which to pay the outflows $X_{1}^{\mathcal{L}}=10$ and to subsequently pay back the borrowed tradables with the capital $C_1$ raised at $t=1$, for which the capital investors are compensated with the excess of the inflows $\tilde{Z}_{2}^{\mathcal{L}}=100$ over the outflows $X_{2}^{\mathcal{L}}=0$ at $t=2$.

\subsection{Production strategies covering failure}\label{ss-prodstratinclfailure}

Production strategies $\phi$ for liabilities $\mathcal{L}$ are defined in Definition~\ref{def-admissibleinvestmentstrategy} for arbitrary fulfillment conditions, so fulfillment may only be required in sufficiently many cases ($A^{\prime}_i\geq L_i$) and not all cases as for full replication. For the following, we assume that the production cost $\Bar{v}_{i}^{\phi, \psi, \mathcal{C}}(\mathcal{L})$ in $L_i$ has been calculated for all cases, including $A^{\prime}_i< L_i$.  Proposition~\ref{prop-completeproductionstrategy} in Section~\ref{ss-prodstratinclfailuredetail} in the Appendix shows that, for a positively homogeneous financiability condition, any production strategy $\phi$ for given liabilities $\mathcal{L}$ with an arbitrary fulfillment condition can be extended, by covering the cases in which there is failure $A^{\prime}_i< L_i$, to a production strategy $\widetilde{\phi}$ for the full fulfillment condition for adjusted liabilities $\widetilde{\mathcal{L}}$ that are identical to the non-adjusted liabilities $\mathcal{L}$ in ``sufficiently many cases". The adjusted liabilities $\widetilde{\mathcal{L}}$ can thus be seen as a ``redefinition" of $\mathcal{L}$ that is ``allowed" by the applicable fulfillment condition. 

However, by the assumptions from Section~\ref{ss-nonnegprodstrategydefn}, balance sheet insolvency $A^{\prime}_i< L_i$ implies failure, so insolvency laws take over. So, the procedure only works in reality if the ``redefinition" is consistent with applicable insolvency laws. The requirement on the insolvency laws is that, in case of failure, there is in particular no bankruptcy liquidation but instead all outstanding cash inflows and outflows are proportionally reduced by essentially the same factor such that balance sheet insolvency is just removed, leading to adjusted liabilities. To avoid breaking the flow, the precise definitions and analyses are provided in Section~\ref{ss-prodstratinclfailuredetail} in the Appendix.

\subsection{Constructing production strategies}\label{ss-constrprodstrat}

Production strategies can be constructed recursively backward in time over successive one-year periods. In this paper, we do not systematically study the existence and construction of production strategies but only provide general comments and specific examples.

In the recursion at date $i\in\{0,1,\ldots , T-1\}$, a production strategy $(\phi_t)_{t=i+1}^{T}$ for the period from $i+1$ to $T$ with capital amounts $\mathcal{C}=(C_t)_{t=i+1}^{T-1}$ and production cost $\Bar{v}_{i+1}^{\phi, \psi,\mathcal{C}}(\mathcal{L})=\phi_{\gamma (i+1)}\cdot S_{i+1}-C_{i+1}$ has been constructed. This may potentially consist of several production strategies defined on disjoint sets in $\mathcal{F}_{i+1}$ and may in practice potentially be only defined for ``sufficiently many cases". With this, the value $L_{i+1}=X_{i+1}^{\mathcal{L}}+\Bar{v}_{i+1}^{\phi, \psi,\mathcal{C}}(\mathcal{L})$ of the liabilities is given. 

The crucial task then is to construct a production strategy $\theta =(\theta_t)_{t=i}^{i+1}$ from date $i$ to $i+1$. With such a strategy $\theta$, first, at dates $t\in D$ with $i<t<i+1$, equation \eqref{def-generalfinancing} must be fulfilled; in particular, the liability cash flows $X_{t}^{\mathcal{L}}$ must be paid. Second, at date $i+1$, with the resulting assets with value $A^{\prime}_{i+1}= \theta_{i+1}\cdot S_{i+1} + \tilde{Z}_{i+1}^{\theta+\psi+\mathcal{L}}$, the fulfillment condition needs to be satisfied, i.e. $A^{\prime}_{i+1}\geq L_{i+1}$ ``in sufficiently many cases". If these two requirements are not satisfied, it can potentially be achieved by increasing the strategy $\theta$, e.g. by adding cash at date $i$ or increasing the units of the tradables. Once the requirements are satisfied, the properties of the financiability condition from Definition~\ref{def-financiabilitycondition} imply that the capital investment $C_i\rightarrow C_{i+1}^{\prime}=(A^{\prime}_{i+1}- L_{i+1})_{+}\geq 0$ satisfies the financiability condition at least for $C_i=0$ and possibly for a maximal $C_i>0$. The production cost at date $i$ are then given by $\Bar{v}_{i}^{\theta, \psi,\mathcal{C}}(\mathcal{L})=\theta_{\gamma (i)}\cdot S_{i}-C_{i}$, and the two strategies can be ``patched together" to a strategy $(\phi_t)_{t=i}^{T}$ from $i$ to $T$.

\subsection{General production strategies}\label{subsec-extensionofadmstrategies}
Definition~\ref{defn-admgeneralstrategies} below extends non-negative production strategies $\phi\in\mathcal{R}^{\geq 0}$ from Definition~\ref{def-admissibleinvestmentstrategy} under specific assumptions to what we call general production strategies $\phi\in\mathcal{R}^{\prime}$, so also allowing for short position and with non-negative values. This extension is in particular needed for Theorem~\ref{prop-valueextension}. Writing $\phi\in\mathcal{R}^{\prime}$ as $\phi=\phi^{+}-\phi^{-}$ with $\phi_t^{\pm}=(\pm\phi_t)_{+}\geq 0$, the extension is defined by interpreting $\phi$ as the non-negative production strategy $\phi^{+}$ for the liabilities $\mathcal{L}+\mathcal{L}^*({\phi^{-}})$, with $\mathcal{L}^*({\phi^{-}})$ as in Definition~\ref{def-liabclosedtakenon}. This requires suitable short positions to be available for production with close out (Definition~\ref{def-negativeshares}). As a simplification, we restrict to full production, i.e. the full fulfillment condition. For more general fulfillment conditions, the liability $\mathcal{L}^*({\phi^{-}})$ is potentially not fulfilled almost surely, so we would have short positions of tradables that are only fulfilled ``in sufficiently many cases". Proposition~\ref{prop-admgeneralstrategies} below shows that, under the above assumptions, general production strategies $\phi\in\mathcal{R}^{\prime}$ are characterized algebraically by the same conditions \eqref{cond-genfinancing}, \eqref{FinCati}, and \eqref{FulCati} as non-negative production strategies in Definition~\ref{def-admissibleinvestmentstrategy}. 
\begin{defn}[General production strategies]\label{defn-admgeneralstrategies}
Let dates $i_{min}<i_{max}$ in $\{0,1,\ldots ,T\}$, liabilities $\mathcal{L}$ with $\mathcal{F}_{i_{max}}$-measurable terminal value $Y_{i_{max}}$, the full fulfillment condition, a financiability condition, and a set $M\in \mathcal{F}_{i_{min}}$ be given. Assume that portfolios containing negative units in $\mathcal{R}^{\prime}$ are available for production with close out (Definition~\ref{def-negativeshares}). A \emph{general production strategy} $\phi$ for producing the liabilities $\mathcal{L}$ with terminal value $Y_{i_{max}}$ on $D\equiv D[i_{min},i_{max}]$ and $M$ using capital and illiquid assets is a strategy $(\phi)_{i=i_{min}}^{i_{max}}\in \mathcal{R}^{\prime}$ with $\phi_t^{\pm}=(\pm\phi_t)_{+}\geq 0$ together with capital amounts $\mathcal{C}=(C_{i})_{i=i_{min}}^{i_{max}}$ with $C_i=C_i(\phi,\psi)\geq 0$ and illiquid assets $\psi$ such that $\phi^{+}\in\mathcal{R}^{\geq 0}$ is a non-negative production strategy with capital amounts $\mathcal{C}$ and illiquid assets $\psi$ for the sum $\mathcal{L}+\mathcal{L}^*({\phi^{-}})$ (Definition~\ref{def-liabclosedtakenon}) with terminal value $Y_{i_{max}}$. The value $\bar{v}_i^{\phi,\psi,\mathcal{C}}(\mathcal{L})$ for $i\in\{i_{min},\ldots ,i_{max}-1\}$ is
\begin{equation}\label{def-valuewithgenstrategy}
	\bar{v}_i^{\phi,\psi,\mathcal{C}}(\mathcal{L})=\bar{v}_i^{\phi^{+},\psi,\mathcal{C}}(\mathcal{L}+\mathcal{L}^*({\phi^{-}}))-v_i(\phi^{-})
\end{equation}	
\end{defn}
\begin{prop}\label{prop-admgeneralstrategies}
Let the full fulfillment condition apply and portfolios containing negative units in $\mathcal{R}^{\prime}$ be available for production with close out (Definition~\ref{def-negativeshares}). A strategy $\phi\in\mathcal{R}^{\prime}$ is a general production strategy for the liabilities $\mathcal{L}$ if and only if the conditions \eqref{cond-genfinancing}, \eqref{FinCati} and \eqref{FulCati} from Definition~\ref{def-admissibleinvestmentstrategy} hold (formally) for $\phi$ and $\mathcal{L}$. 
\end{prop}
\begin{proof}
Let $\mathcal{L}^{*}\equiv \mathcal{L}^{*}({\phi^{-}})$. Because of the full fulfillment condition, conditioning on $A_i^{\prime}\geq L_i$ in Definition~\ref{def-admissibleinvestmentstrategy} becomes redundant. For dates $t\in D\cap\interval[open]{i}{i+1}$ with $i\in\{i_{min},\ldots, i_{max}-1\}$, the equivalence of the equation corresponding to \eqref{def-generalfinancing} from condition \eqref{cond-genfinancing} for $\phi^{+}$ and $\mathcal{L}+\mathcal{L}^{*}$ and for $\phi$ and $\mathcal{L}$ is a special case of the calculation in Section~\ref{ss-genstratrestrilliquid} around \eqref{eqn-genstrat} by using the definitions of $X_{t}^{\mathcal{L}^*}$ and $\tilde{Z}_{t}^{\mathcal{L}^{*}}$ of the cash flows from Definition~\ref{def-liabclosedtakenon} and setting $\tilde{Z}_{t}=\tilde{Z}_{t}^{\psi +\mathcal{L}}$ and $X_{t}=X_{t}^{\mathcal{L}}$.
	
Conditions \eqref{FinCati} and \eqref{FulCati} follow from showing recursively backward in time that the expression $A_{i}^{\prime}-L_{i}$ for $i\in\{0,\ldots , i_{max}\}$ for $\phi^{+}$ and $\mathcal{L}+\mathcal{L}^{*}$, which is $A_{i}^{\prime}-L_{i}=\phi_i^{+}\cdot S_i+\tilde{Z}_i^{\phi^{+}+\psi+\mathcal{L}+\mathcal{L}^*}-X_i^{\mathcal{L}+\mathcal{L}^*}-\bar{v}_i^{\phi^{+},\psi,\mathcal{C}}(\mathcal{L}+\mathcal{L}^*$), is the same expression as for $\phi$ and $\mathcal{L}$. This follows for $i=i_{max}$ and $i<i_{max}$ from Definition~\ref{def-liabclosedtakenon}, the expressions for $\bar{v}_{i}^{\phi,\psi,\mathcal{C}}(\mathcal{L})$ in Definition~\ref{defn-admgeneralstrategies}, and \eqref{def-valuestrategy}.
\end{proof}

\section{Properties of the valuation}\label{sec-propertiesofvalue}

\subsection{Valuation in the framework as extension of market price valuation}\label{subsec-perfectreplication}

Theorem~\ref{prop-valueextension} below provides the crucial proof that, under suitable assumptions, specifically on the financiability condition, the valuation framework for liabilities introduced in this paper is an extension of the valuation of tradable investment strategies $\phi$ by the market price $v_i(\phi)$ as in \eqref{def-valuestrategy}. Similar to ``classical no-arbitrage valuation", the setting is full production (Definition~\ref{def-fulfillmentcondition}), illiquid assets $\psi =0$, and general production strategies (Definition~\ref{defn-admgeneralstrategies}), and assuming that portfolios containing negative units in $\mathcal{R}^{\prime}$ are available for production with close out (Definition~\ref{def-negativeshares}). Unlike ``classical no-arbitrage valuation", we additionally allow for capital investments. We restrict to non-negative investment strategies $\phi\in\mathcal{R}^{\geq 0}$ for simplicity and show that a suitable liability corresponding to $\phi$ can broadly be produced with production cost $v_i(\phi)$ and no lower production cost. We do this for the liability $\mathcal{L}(\phi)$ defined in Definition~\ref{def-negativeshares} with terminal value $Y_{T} =0$. Capital investments can be seen as financial instruments in addition to tradables, and it is shown in Theorem~\ref{prop-valueextension} that production cost can only be an extension of the market price if the financiability condition is consistent with the tradables in the sense of Definition~\ref{def-financiabilitycondition}.

The proof of Theorem~\ref{prop-valueextension} uses the following two lemmas. The first lemma examines the natural production strategy for $\mathcal{L}(\phi)$, which is $\phi$ itself with an adjustment for the stopping time $\tau$.
\begin{lemma}\label{lemma-naturalprodstrategy}
Let the full fulfillment condition apply, portfolios containing negative units in $\mathcal{R}^{\prime}$ be available for production with close out (Definition~\ref{def-negativeshares}), and illiquid assets $\psi=0$. For a strategy $\phi\in\mathcal{R}^{\geq 0}$ on $D\equiv D[0,T]$ with zero cash inflows $\tilde{Z}_t=0$ and cash outflows $X_t$, let $\mathcal{L}(\phi)$ be the liability from Definition~\ref{def-negativeshares}. Define the strategy $\phi^{\prime}\in\mathcal{R}^{\geq 0}$ for $t\in D$ by $\phi_t^{\prime}=\phi_t\cdot 1_{\{t\leq\tau\}}$ and $\phi_{\gamma (T)}^{\prime}=0$. Then, for any $t\in D$, the cash outflows $X_t^{\prime}$ of $\phi^{\prime}$ are $X_t^{\prime}=X_t^{\mathcal{L}(\phi)}$, i.e. $\phi_{\gamma (t)}^{\prime}\cdot S_t=\phi_{t}^{\prime}\cdot S_t+\tilde{Z}_t^{\phi^{\prime}}-X_t^{\mathcal{L}(\phi)}$. Further, the strategy $\phi^{\prime}$ with capital $C_i=0$ is a production strategy for the liability $\mathcal{L}(\phi)$ with terminal value $Y_T=0$, which has production cost $\bar{v}_i^{\phi^{\prime},0,0}(\mathcal{L}(\phi))=v_i(\phi^{\prime})=v_i(\phi)\cdot 1_{\{t<\tau\}}$. 
\end{lemma}
\begin{proof}
By Definition~\ref{def-negativeshares}, we have $X_t^{\mathcal{L}(\phi)}=1_{\{t<\tau\}}\cdot X_t+1_{\{t=\tau\}}\cdot (\phi_t\cdot S_t+\tilde{Z}_t^{\phi})$ for $t\in D$. We first show that the cash flows $X_t^{\prime}$ are given by the same expression. This holds for $t<\tau$ because then $\gamma (t)\leq\tau$, so the definition of $\phi^{\prime}$ and \eqref{eqn-selffinancing} for $\phi$ immediately imply $\phi_{\gamma (t)}^{\prime}\cdot S_t=\phi_{\gamma (t)}\cdot S_t=\phi_{t}\cdot S_t+\tilde{Z}_t^{\phi}-X_t=\phi_{t}^{\prime}\cdot S_t+\tilde{Z}_t^{\phi^{\prime}}-X_t$, so $X_t^{\prime}=X_t$. For $t=\tau$, we have $\gamma (t)>\tau$ and so the definition of $\phi^{\prime}$ implies $\phi_{\gamma (t)}^{\prime}\cdot S_t=0=\phi_{t}^{\prime}\cdot S_t+\tilde{Z}_t^{\phi^{\prime}}-(\phi_{t}\cdot S_t+\tilde{Z}_t^{\phi})$. For $t>\tau$, the argument is similar. Thus, $\phi_{\gamma (t)}^{\prime}\cdot S_t=\phi_{t}^{\prime}\cdot S_t+\tilde{Z}_t^{\phi^{\prime}}-X_t^{\mathcal{L}(\phi)}$ for any $t\in D$. In particular, \eqref{def-generalfinancing} from Definition~\ref{def-admissibleinvestmentstrategy} holds for $\phi^{\prime}$ for any $t\in D\setminus [0,\ldots ,T]$. It further implies for $i\in [0,\ldots ,T-1]$ that $A_{i+1}^{\prime}-L_{i+1}=\phi_{i+1}^{\prime}\cdot S_{i+1}+\tilde{Z}_{i+1}^{\phi^{\prime}}-X_{i+1}^{\mathcal{L}(\phi)}-\bar{v}_{i+1}^{\phi^{\prime},0,0}(\mathcal{L}(\phi))=v_{i+1}(\phi^{\prime})-\bar{v}_{i+1}^{\phi^{\prime},0,0}(\mathcal{L}(\phi))$. For $i=T-1$, this is zero because $\phi_{\gamma (T)}^{\prime}=0$ and $\bar{v}_T^{\phi^{\prime},0,0}(\mathcal{L}(\phi))=Y_T=0$ by assumption. So the full fulfillment condition and the financiability condition with $C_i=0$ hold. It follows that the production cost at $i=T-1$ are $\bar{v}_i^{\phi^{\prime},0,0}(\mathcal{L}(\phi))=v_i(\phi^{\prime})-C_i=v_i(\phi^{\prime})$. Thus, $A_{i+1}^{\prime}-L_{i+1}=0$ at $i=T-2$. The proof proceeds recursively. 
\end{proof}
The next lemma shows that any production cost for the liability $\mathcal{L}(\phi)$ is at least as large as $v_i(\phi^{\prime})$ from Lemma~\ref{lemma-naturalprodstrategy}.
\begin{lemma}\label{lemma-prodstrategynonnegandcloseout}
Consider the situation from Lemma~\ref{lemma-naturalprodstrategy}. Assume consistency within the tradables in $\mathcal{R}$ and with the financiability condition. Let $\theta\in\mathcal{R}^{\prime}$ with capital amounts $\mathcal{C}=(C_i)_{i=0}^{T-1}$ be a general production strategy on $D$ for the liability $\mathcal{L}(\phi)$. Then, for any $t\in D$ and $i\in\{0,\ldots , T\}$, 
\begin{equation}\label{eqn-twoclaims}
	v_t(\theta)\geq v_t(\phi)\cdot 1_{\{t<\tau\}}\mbox{, and }\Bar{v}_{i}^{\theta,0,\mathcal{C}}(\mathcal{L}(\phi))\geq v_{i}(\phi)\cdot 1_{\{i<\tau\}}
\end{equation}
\end{lemma}
\begin{proof}
We show the two expressions \eqref{eqn-twoclaims} by applying consistency with the financiability condition from Definition~\ref{def-financiabilitycondition} recursively backward in time to the strategy $\xi =\theta - \phi^{\prime}$ with $\phi_t^{\prime}=\phi_t\cdot 1_{\{t\leq\tau\}}$ from Lemma~\ref{lemma-naturalprodstrategy}. For $i\in \{0,\dots ,T-1\}$, because of the full fulfillment condition, $0\leq C_{i+1}^{\prime, \theta}=A_{i+1}^{\prime ,\theta}-L_{i+1}^{\theta}=\theta_{i+1}\cdot S_{i+1}+\tilde{Z}^{\theta}_{i+1}-X_{i+1}^{\mathcal{L}(\phi)}-\Bar{v}_{i+1}^{\theta,0,\mathcal{C}}(\mathcal{L}(\phi))$. Inserting the expression for $X_{i+1}^{\mathcal{L}(\phi)}$ derived from $v_{i+1}(\phi^{\prime})=\phi_{\gamma ({i+1})}^{\prime}\cdot S_{i+1}=\phi_{i+1}^{\prime}\cdot S_{i+1}+\tilde{Z}_{i+1}^{\phi^{\prime}}-X_{i+1}^{\mathcal{L}(\phi)}$ by Lemma~\ref{lemma-naturalprodstrategy}, we can write this as 
\begin{equation}\label{equinlemmanonneg}
0\leq C_{i+1}^{\prime, \theta}=\xi_{i+1}\cdot S_{i+1}+\tilde{Z}^{\xi}_{i+1}+v_{i+1}(\phi^{\prime})-\Bar{v}_{i+1}^{\theta,0,\mathcal{C}}(\mathcal{L}(\phi))
\end{equation}
For $i=T-1$, \eqref{equinlemmanonneg} reduces to $0\leq C_{i+1}^{\prime, \theta}=\xi_{i+1}\cdot S_{i+1}+\tilde{Z}^{\xi}_{i+1}$ or equivalently $\phi^{\prime}_{i+1}\cdot S_{i+1}+\tilde{Z}^{\phi^{\prime}}_{i+1}\leq \theta_{i+1}\cdot S_{i+1}+\tilde{Z}^{\theta}_{i+1}$. As $\phi^{\prime}\in\mathcal{R}^{\geq 0}$ (by Lemma~\ref{lemma-naturalprodstrategy}) and $\theta\in\mathcal{R}^{\prime}$ are production strategies for $\mathcal{L}(\phi)$, they have the same cash flows for any $t\in D\setminus\{0,\dots ,T\}$, so $\xi$ is self-financing. Thus, we can apply consistency of tradables \eqref{eqn-consistencyconditionsimple} successively backward to get for any $i\leq t<i+1$ that $v_{t}(\phi^{\prime})\leq v_{t}(\theta)$, i.e. $v_t(\xi)\geq 0$. Thus, we can apply consistency with the tradables as in Definition~\ref{def-financiabilitycondition} to conclude that $C_i\leq v_i(\xi)=v_i(\theta)-v_i(\phi^{\prime})$ and thus $\bar{v}_i^{\theta,0,\mathcal{C}}(\mathcal{L}(\phi))=v_i(\theta)-C_i\geq v_i(\phi^{\prime})=v_{i}(\phi)\cdot 1_{\{i<\tau\}}$. Inserting this into \eqref{equinlemmanonneg} with $i+1$ replaced by $i$ gives $0\leq C_{i}^{\prime, \theta}=\xi_{i}\cdot S_{i}+\tilde{Z}^{\xi}_{i}$ for $i=T-1$. The proof proceeds recursively.
\end{proof}
\begin{theorem}\label{prop-valueextension}
Let the full fulfillment condition apply, portfolios containing negative units in $\mathcal{R}^{\prime}$ be available for production with close out, and illiquid assets $\psi=0$. Assume consistency within the tradables in $\mathcal{R}$ (Definition~\ref{def-financiabilitycondition}). Then the following are equivalent:
\begin{enumerate}[(a)]
    \item\label{equ-valuerecovered} For any strategy $\phi\in\mathcal{R}^{\geq 0}$ on $D[0,T]$ with zero cash inflows $\tilde{Z}_t=0$ and cash outflows $X_t$, the value of the liability $\mathcal{L}(\phi)$ with $\bar{v}_T(\mathcal{L}(\phi))=0$ is $\Bar{v}_i(\mathcal{L}(\phi))=v_i(\phi)\cdot 1_{\{i<\tau\}}$ for $i\in\{0,\ldots , T-1\}$.
		\item\label{cond-finc} The financiability condition is consistent with the tradables.
\end{enumerate}
\end{theorem}
\begin{proof}
\eqref{equ-valuerecovered} follows from \eqref{cond-finc} because of Lemmas~\ref{lemma-naturalprodstrategy} and \ref{lemma-prodstrategynonnegandcloseout}. We show that \eqref{cond-finc} follows from \eqref{equ-valuerecovered} by contraposition: If \eqref{cond-finc} does not hold, there is an $i\in\{0,\ldots , T-1\}$, a self-financing strategy $\theta\in\mathcal{R}^{\prime}$ from $i$ to $i+1$, capital $C_{i}$ with $C_{i}>v_{i}(\theta)$ on a set $M\in\mathcal{F}_{i}$ with positive probability, and $C_{i}\rightarrow \theta_{i+1}\cdot S_{i+1}+\tilde{Z}_{i+1}^{\theta}$ satisfies the financiability condition. To show that \eqref{equ-valuerecovered} does not hold, select the strategy $\phi=0$ in $\mathcal{R}^{\geq 0}$, so $\mathcal{L}(\phi)=0$ and $v_i(\phi)=0$. We define the strategy $\xi$ for $t\in D[0,T]$ by $\xi_t=0$ for $t\notin\interval{i}{i+1}$ as well as for $t\in\interval{i}{i+1}$ on $\Omega\setminus M$, and by $\xi_t=\theta_t$ for $t\in\interval{i}{i+1}$ on $M$. Then, from $i+1$ to $T$, the strategy $\xi=0$ with capital $C_{j}=0$ is a production strategy for $\mathcal{L}(\phi)=0$, so in particular, $\bar{v}_{i+1}^{\xi,0,0}(\mathcal{L}(\phi))=0$. From $i$ to $i+1$, consider $\xi$ together with capital $\widetilde{C}_{i}$, defined to be $C_{i}$ on $M$ and zero otherwise, for the liability $\mathcal{L}(\phi)=0$. Then, $C_{i+1}^{\prime}=A_{i+1}^{\prime}-L_{i+1}=1_M\cdot (\theta_{i+1}\cdot S_{i+1}+\tilde{Z}_{i+1}^{\theta})\geq 0$ as $X_{i+1}^{\mathcal{L}(\phi)}+\bar{v}_{i+1}^{\xi,0,0}(\mathcal{L}(\phi))=0$, so the strategy $\xi$ with capital $\widetilde{C}_{i}\geq 0$ is a production strategy for $\mathcal{L}(\phi)=0$ and thus the value $\Bar{v}_{i}(0)\leq v_{i}(\xi)-\widetilde{C}_{i}$. Specifically on $M$, we have $\Bar{v}_{i}(0)\leq v_{i}(\theta)-C_{i} < 0 = v_{i}(0)$, so \eqref{equ-valuerecovered} does not hold.
\end{proof}
\begin{remark}\label{rem-strictconsistencywithfincond}
Requiring a financiability condition to be consistent with the tradables as in Definition~\ref{def-financiabilitycondition} does not preclude that investments in self-financing strategies of tradables satisfy the financiability condition. If self-financing strategies of tradables satisfy the financiability condition, the proof of Lemma~\ref{lemma-prodstrategynonnegandcloseout} can be used to show that any production strategy $\theta\in\mathcal{R}^{\prime}$ with capital amounts $\mathcal{C}=(C_i)_{i=0}^{T-1}$ for the liability $\mathcal{L}(\phi)$, potentially adjusted by replacing $\mathcal{C}$ with lower capital amounts $\mathcal{C}^{*}$, has ``minimal" production cost equal to the market price, i.e. $\Bar{v}_{i}^{\theta,0,\mathcal{C}^{*}}(\mathcal{L}(\phi))= v_{i}(\phi)\cdot 1_{\{i<\tau\}}$. Intuitively speaking, this means in particular that capital can be raised, even when it would not be required, without increasing the production cost. It can be argued that this should be prevented by making the financiability condition stricter. That is, such that capital investments require a higher return than the return from self-financing strategies (except of course for the investment $0\rightarrow 0$). This may also be plausible given the additional risks investors may face from investing in an insurance company instead of in tradables. In this case, the financiability condition cannot also be neutral to the tradables according to Definition~\ref{def-financiabilitycondition}, at least not for $C_i=0$ and $C_{i+1}^{\prime}=0$.
\end{remark}
Theorem~\ref{prop-valueextension} is limited to perfect fulfillment. With a weaker fulfillment condition, we can for example not in general expect that the liability $\mathcal{L}(\phi)=0$ cannot be produced with a production strategy with negative production cost. For perfect fulfillment, on the other hand, further normative properties can be expected to hold. In particular, the following corollary shows that a liability with ``non-negative cash flows" has a non-negative value.    
\begin{corollary}\label{cor-nonegcashflows}
Let the full fulfillment condition apply and portfolios containing negative units in $\mathcal{R}^{\prime}$ be available for production with close out. Assume consistency within the tradables in $\mathcal{R}$ and with the financiability condition. Let the liability $\mathcal{L}$ with terminal value $Y_T=0$ and the illiquid assets $\psi$ be such that ``cash flows are non-negative" in the sense that, for all $t\in D\equiv  D[0,T]$,
\begin{equation}
	X_t^{\mathcal{L}}-\tilde{Z}_t^{\psi+\mathcal{L}} \geq 0
\end{equation}
Then the liability $\mathcal{L}$ has non-negative value, $\Bar{v}_{i}^{\psi}(\mathcal{L})\geq 0$ for $i\in \{0,\ldots ,T-1\}$.
\end{corollary}
\begin{proof}
Given an arbitrary production strategy $\phi\in\mathcal{R}^{\prime}$ for $\mathcal{L}$, we construct a production strategy $\theta = \phi+\xi\in\mathcal{R}^{\prime}$ with $\xi\in\mathcal{R}^{\geq 0}$, capital $\mathcal{C}(\theta)$, and no illiquid assets for the liability $\mathcal{L}(\phi)=0$ with terminal value $0$ such that $\Bar{v}_{i}^{\phi, \psi,\mathcal{C}}(\mathcal{L})\geq \Bar{v}_{i}^{\theta, 0,\mathcal{C}(\theta)}(0)$ for $i\in \{0,\ldots ,T-1\}$. It then follows from the definition of the value and Theorem~\ref{prop-valueextension} that $\Bar{v}_{i}^{\phi, \psi,\mathcal{C}}(\mathcal{L})\geq \Bar{v}_{i}(0)=0$, so the claim follows from the definition of the value $\Bar{v}_{i}^{\psi}(\mathcal{L})$.

We construct the production strategy $\theta$ recursively backward in time for $i\in \{0,\ldots ,T-1\}$, forward in time within each one-year period from date $i$ to $i+1$, such that $\Bar{v}_{i}^{\phi, \psi,\mathcal{C}}(\mathcal{L})\geq \Bar{v}_{i}^{\theta, 0,\mathcal{C}(\theta)}(0)$. At date $i+1=T$, the inequality holds as $\Bar{v}_{T}^{\phi, \psi,\mathcal{C}}(\mathcal{L})=0= \Bar{v}_{T}^{\theta, 0,\mathcal{C}(\theta)}(0)$ by assumption. From $i$ to $i+1$, we define $\xi_{\gamma (i)}=0$ and, at $\gamma (i)\leq t<i+1$, given $\xi_t$, we select $\xi_{\gamma (t)}\in \mathcal{R}^{\geq 0}$ such that 
\begin{equation}
	\xi_{\gamma (t)}\cdot S_t=\xi_t\cdot S_t+\tilde {Z}_t^{\xi}+X_t^{\mathcal{L}}-\tilde{Z}_t^{\psi+\mathcal{L}}\geq 0
\end{equation}
which is possible as $X_t^{\mathcal{L}}-\tilde{Z}_t^{\psi+\mathcal{L}}\geq 0$ by assumption. Then, $\theta_{\gamma (t)}\cdot S_t = \phi_{\gamma (t)}\cdot S_t+\xi_{\gamma (t)}\cdot S_t$ is non-negative as $\phi_{\gamma (t)}\cdot S_t\geq 0$ and $\xi_{\gamma (t)}\cdot S_t\geq 0$ and, using \eqref{cond-genfinancing} from Definition~\ref{def-admissibleinvestmentstrategy} for $\phi$, equal to
 \begin{equation}
	\phi_t\cdot S_t+\tilde {Z}_t^{\phi+\psi+\mathcal{L}}-X_t^{\mathcal{L}}+\xi_t\cdot S_t+\tilde {Z}_t^{\xi}+X_t^{\mathcal{L}}-\tilde{Z}_t^{\psi+\mathcal{L}}=\theta_t\cdot S_t +\tilde {Z}_t^{\theta}-0
\end{equation}
so that \eqref{cond-genfinancing} from Definition~\ref{def-admissibleinvestmentstrategy} holds for $\theta$. At date $i+1$, we denote the assets and liabilities according to Definition~\ref{def-assetsandliabilities} by $A_{i+1}^{\prime,\phi}$ and $L_{i+1}^{\phi}$ for $\mathcal{L}$ and the strategy $\phi$, and as $A_{i+1}^{\prime,\theta}$ and $L_{i+1}^{\theta}$ for the liability $0$ and the strategy $\theta$. Using the respective definitions, we can write $A_{i+1}^{\prime,\theta} - L_{i+1}^{\theta}$ as
 \begin{equation}
	A_{i+1}^{\prime,\phi} - L_{i+1}^{\phi} + (\xi_{i+1}\cdot S_{i+1}+\tilde {Z}_{i+1}^{\xi}) +( X_{i+1}^{\mathcal{L}}-\tilde{Z}_{i+1}^{\psi+\mathcal{L}}) + (\Bar{v}_{i+1}^{\phi, \psi,\mathcal{C}}(\mathcal{L})- \Bar{v}_{i+1}^{\theta, 0,\tilde{\mathcal{C}}}(0))
\end{equation}
By the various properties and assumptions, $A_{i+1}^{\prime,\theta} - L_{i+1}^{\theta}\geq A_{i+1}^{\prime,\phi} - L_{i+1}^{\phi}\geq 0$, so the full fulfillment condition is satisfied for $\theta$. Hence, also, for the return to the capital provider,  $C_{i+1}^{\prime,\theta}\geq C_{i+1}^{\prime,\phi}$ with the analogue notation, hence by the properties of the financiability condition, $C_{i}^{\theta}\geq C_{i}^{\phi}\equiv C_{i}$. We thus get for the production cost $\Bar{v}_{i}^{\theta, 0,\mathcal{C}(\theta)}(0)$ date $i$, using that $v_i(\theta)=\theta_{\gamma (i)}\cdot S_i=\phi_{\gamma (i)}\cdot S_i=v_i(\phi)$, 
\begin{equation}
	\Bar{v}_{i}^{\theta, 0,\mathcal{C}(\theta)}(0)=v_i(\theta)-C_{i}^{\theta} \leq v_i(\phi) - C_{i} = \Bar{v}_{i}^{\phi, \psi,\mathcal{C}}(\mathcal{L})
\end{equation}
The proof then proceeds recursively.
\end{proof}

\subsection{Adding short positions}\label{ss-addingshortpos}

We show in the following that the production cost is additive when adding a liability $\mathcal{L}(\phi)$ for $\phi\in\mathcal{R}^{\geq 0}$ as in Definition~\ref{def-negativeshares} with cash outflows $X_t$ and zero cash inflows $\tilde{Z}_t=0$ to a liability $\mathcal{L}$ with terminal value $Y_T$. We assume the full fulfillment condition and that portfolios with negative units in $\mathcal{R}^{\prime}$ are available for production with close out (Definition~\ref{def-negativeshares}). Let $\theta\in\mathcal{R}^{\prime}$ with capital $\mathcal{C}=(C_{i})_{0}^{T-1}$ and illiquid assets $\psi$ be a production strategy for the liabilities $\mathcal{L}$. We show that $\theta +\phi^{\prime}$ for $\phi^{\prime}_t=1_{\{t\leq\tau\}}\cdot\phi_t$ with the same capital and illiquid assets is a production strategy for $\mathcal{L}+\mathcal{L}(\phi)$ with the same terminal value $Y_T$ and with production cost
\begin{equation}\label{expr-prodcostliabwithphi}
	\bar{v}_i^{\theta +\phi^{\prime},\psi,\mathcal{C}}(\mathcal{L}+\mathcal{L}(\phi))=\bar{v}_i^{\theta,\psi,\mathcal{C}}(\mathcal{L})+v_i(\phi^{\prime}) 
\end{equation}
By Lemma~\ref{lemma-naturalprodstrategy}, $\phi^{\prime}$ is a production strategy for $\mathcal{L}(\phi)$ with no illiquid assets, so \eqref{def-generalfinancing} holds for $\theta +\phi^{\prime}$ and $\mathcal{L}+\mathcal{L}(\phi)$ by additivity. For $i\in\{0,\ldots ,T-1\}$, at date $i+1$, we have $A_{i+1}^{\prime\theta +\phi^{\prime}}-L_{i+1}^{\theta +\phi^{\prime}}=(\theta +\phi^{\prime})_{i+1}\cdot S_{i+1}+\tilde{Z}_{i+1}^{\theta+\phi^{\prime}+\psi+\mathcal{L}+\mathcal{L}(\phi)}-X_{i+1}^{\mathcal{L}+\mathcal{L}(\phi)}-\bar{v}_{i+1}^{\theta +\phi^{\prime},\psi}(\mathcal{L}+\mathcal{L}(\phi))=\theta_{i+1}\cdot S_{i+1}+\tilde{Z}_{i+1}^{\theta+\psi+\mathcal{L}}-X_{i+1}^{\mathcal{L}}+v_{i+1}(\phi^{\prime})-\bar{v}_{i+1}^{\theta +\phi^{\prime},\psi}(\mathcal{L}+\mathcal{L}(\phi))$ using Lemma~\ref{lemma-naturalprodstrategy}. If \eqref{expr-prodcostliabwithphi} holds for $i+1$ instead of $i$, it follows that $A_{i+1}^{\prime\theta +\phi^{\prime}}-L_{i+1}^{\theta +\phi^{\prime}}=A_{i+1}^{\prime\theta}-L_{i+1}^{\theta}$, and as $\theta$ is a production strategy for $\mathcal{L}$, fulfillment and financiability condition are satisfied from $i$ to $i+1$ and hence $\bar{v}_i^{\theta +\phi^{\prime},\psi}(\mathcal{L}+\mathcal{L}(\phi))=v_i(\theta+\phi^{\prime})-C_i=v_i(\theta)-C_i+v_i(\phi^{\prime})=\bar{v}_i^{\theta,\psi}(\mathcal{L})+v_i(\phi^{\prime})$, so \eqref{expr-prodcostliabwithphi} holds for $i$. The recursion starts at $i+1=T$, where \eqref{expr-prodcostliabwithphi} holds by assumption.

\subsection{Production with illiquid assets}\label{subsec-prodwithilliquidassets}

Recall that we allow for a portfolio $\psi$ of ``fully illiquid" assets in the production, which we assume are ``held to maturity" (HTM). In particular, in contrast to a strategy of tradables, no additional illiquid assets can be bought, and only the cash inflows $\tilde{Z}_t^{\psi}$ of the illiquid assets at the given date $t$ can be used in the production and not the proceeds from selling part of the assets. Similarly, cash flows of the illiquid assets after date $T$ have an impact on the production only potentially when the assets are transferred at date $T$ to the capital investors, for whom their value is likely less than if the cash flows came from tradables. Holding assets to maturity thus potentially restricts the range of available production strategies, so a priori, the resulting value of the liabilities is at least as large as when trading assets is possible. 

The treatment of the illiquid assets in the production is identical to how cash inflows (premiums) from the liabilities $\mathcal{L}$ are treated. For this reason, alternatively, we could view the object to be produced to be the liabilities $\mathcal{L}$ together with the illiquid assets $\psi$, effectively by adding the illiquid asset cash flows to the premium cash inflows. A similar type of ``cash flow offsetting" that is common in practice is to value the insurance liabilities ``net" of the cash flows of the outgoing reinsurance covers. 

In the balance sheet value $A_i^{\prime}$ of the assets at a date $i$ according to Definition~\ref{def-assetsandliabilities}, no value is assigned to the future cash flows of the illiquid assets $\psi$, in line with the approach for production. At the same time, these future cash flows plausibly reduce the production cost on the liability side. By how much? One can conjecture that this reduction is equal to the value $v_i(\psi)$ of the outstanding cash flows in the case that $\psi$ is a suitable ``static" investment strategy of tradables instead of illiquid assets, i.e. $\psi_{\gamma (t)}=\psi_t$ for any $t\in D$, 
 cash flows $X_t=\tilde{Z}_t^{\psi}$, no cash inflows, and, for simplicity, terminal value $v_T(\psi)=0$. The following proposition shows this for financiability conditions that are neutral to the tradables and for general fulfillment conditions and strategies in $\mathcal{R}^{\geq 0}$. The proposition can alternatively be formulated for the full fulfillment condition and strategies in $\mathcal{R}^{\prime}$ using suitable assumptions. Note that, for tradables of value $v_i(\psi)$ at date $i$, a different production strategy with lower production cost might be available. 
\begin{prop}
Let the financiability condition be neutral to the tradables (Definition~\ref{def-financiabilitycondition}). Given liabilities $\mathcal{L}$ on $D\equiv D[0,T]$ with terminal value $Y_T$, let $\phi\in\mathcal{R}^{\geq 0}$ with capital $\mathcal{C}=(C_i)_{i=0}^{T-1}$ be a production strategy for the liabilities $\mathcal{L}$ with no illiquid assets and production cost $\bar{v}_i^{\phi,0,\mathcal{C}}(\mathcal{L})$. Let $\psi\in\mathcal{R}^{\geq 0}$ be an investment strategy of tradables with cash outflows $X_t=\tilde{Z}_t^{\psi}$, no cash inflows, $\psi_{\gamma (t)}=\psi_t\geq 0$ on $D$, and $v_T(\psi)=0$. For $i\in\{0,\ldots , T-1\}$, let the investment strategy $\xi\equiv \xi^{(i)}\in\mathcal{R}^{\geq 0}$ of tradables be defined for $t\in D$ with $i<t<i+1$ such that $\xi_{\gamma(i)}=0$ and $\xi_{\gamma(t)}\cdot S_t=\xi_t\cdot S_t+\tilde{Z}_t^{\xi+\psi}$ for $i<t<i+1$. Then, the strategy $\phi +\xi$ with capital $\mathcal{C}^{*}=(C_i+v_i(\psi))_{i=0}^{T-1}$ and ``illiquid assets" $\psi$ is a production strategy for $\mathcal{L}$ with production cost 
\begin{equation}
\bar{v}_i^{\phi +\xi,\psi,\mathcal{C}^{*}}(\mathcal{L})=\bar{v}_i^{\phi,0,\mathcal{C}}(\mathcal{L})-v_i(\psi)
\end{equation}
\end{prop}
\begin{proof}
Let $i\in\{0,\ldots ,T-1\}$. For $\phi +\xi\in\mathcal{R}^{\geq 0}$, condition \eqref{def-generalfinancing} holds at $t\in D$ with $i<t<i+1$ because, using \eqref{def-generalfinancing} for $\phi$, we have $(\phi +\xi)_t\cdot S_t+\tilde{Z}_t^{(\phi+\xi)+\psi+\mathcal{L}}-X_t^{\mathcal{L}}=(\phi_t\cdot S_t+\tilde{Z}_t^{\phi+\mathcal{L}}-X_t^{\mathcal{L}})+(\xi_t\cdot S_t+\tilde{Z}_t^{\xi+\psi})=(\phi+\xi)_{\gamma (t)}\cdot S_t$. We denote the value of the assets and liabilities at date $i+1$ by $A_{i+1}^{\prime}$ and $L_{i+1}$ for $\phi$, and by $A_{i+1}^{*\prime}$ and $L_{i+1}^{*}$ for $\phi+\xi$. Using that $\psi_{i+1}\cdot S_{i+1}=v_{i+1}(\psi)$ by assumption in the following last equality sign, we get $A_{i+1}^{*\prime}-L_{i+1}^{*}=(\phi+\xi)_{i+1}\cdot S_{i+1}+\tilde{Z}_{i+1}^{\phi +\xi +\psi+\mathcal{L}}-X_{i+1}^{\mathcal{L}}-\bar{v}_{i+1}^{\phi +\xi,\psi,\mathcal{C}^{*}}(\mathcal{L})=A_{i+1}^{\prime}-L_{i+1}+(\xi+\psi)_{i+1}\cdot S_{i+1}+\tilde{Z}_{i+1}^{\xi +\psi}+\bar{v}_{i+1}^{\phi,0,\mathcal{C}}(\mathcal{L})-\bar{v}_{i+1}^{\phi +\xi,\psi,\mathcal{C}^{*}}(\mathcal{L})-v_{i+1}(\psi)$. We proceed recursively backward in time and show that the last three terms together are equal to zero, i.e. $A_{i+1}^{*\prime}-L_{i+1}^{*}=A_{i+1}^{\prime}-L_{i+1}+(\xi+\psi)_{i+1}\cdot S_{i+1}+\tilde{Z}_{i+1}^{\xi +\psi}$. For $i=T-1$, this holds because $\bar{v}_{T}^{\phi,0,\mathcal{C}}(\mathcal{L})=Y_T=\bar{v}_{T}^{\phi +\xi,\psi,\mathcal{C}^{*}}(\mathcal{L})$ and $v_{T}(\psi)=0$ by assumption. Then, because $(\xi+\psi)_{i+1}\cdot S_{i+1}+\tilde{Z}_{i+1}^{\xi +\psi}\geq 0$, the fulfillment condition holds for $\phi +\xi$ and for the capital payoff $C_{i+1}^{*\prime}\geq C_{i+1}^{\prime}+(\xi+\psi)_{i+1}\cdot S_{i+1}+\tilde{Z}_{i+1}^{\xi +\psi}$. The strategy $\xi +\psi\in\mathcal{R}^{\geq 0}$ is self-financing by construction: for $i<t<i+1$, we have $(\xi+\psi)_{\gamma (t)}\cdot S_t=(\xi +\psi)_t\cdot S_t+\tilde{Z}_t^{\xi +\psi}$. Neutrality of the financiability condition to the tradables then implies that $C_i^{*}=C_i+v_i(\xi+\psi)\rightarrow C_{i+1}^{*\prime}$ satisfies the fulfillment condition. Thus, for the production cost at date $i$, we get $\bar{v}_i^{\phi +\xi,\psi,\mathcal{C}^{*}}(\mathcal{L})=v_i(\phi+\xi)-C_i^{*}=v_i(\phi+\xi)-C_i-v_i(\xi+\psi)=\bar{v}_i^{\phi,0,\mathcal{C}}(\mathcal{L})-v_i(\psi)$ using $v_i(\xi)=0$ as $\xi_{\gamma (i)}=0$. Hence, we can proceed recursively. 
\end{proof}

\section{Valuation of insurance liabilities under Solvency II and in the SST}\label{ss-valuesolvencyiisst}

In the following we set out a way to derive a valuation of insurance liabilities corresponding to Solvency II (\cite{EUFD2009}, \cite{EUDR2014}) and the SST (\cite{AVOSST}) from the framework. This involves specifying the applicable fulfillment and financiability condition and production strategies, as well as additional simplifying assumptions.

{\bf Solvency II and SST.} We recall that an insurance company is solvent at a reference date $i=0$ if the available capital $AC_0$ (Eligible Own Funds in Solvency II, Risk-Bearing Capital in the SST) is at least as large as the regulatory required solvency capital $SCR_0$ (Solvency Capital Requirement in Solvency II, Target Capital in the SST):
\begin{equation}\label{eq-avandreqcap}
AC_0\geq SCR_0
\end{equation}
The available capital $AC_0$ is roughly the value of the assets minus the value of the liabilities. The required capital $SCR_0$ is largely derived from the discounted one-year change in available capital: 
\begin{equation}\label{eq-regreqsolvcap}
SCR_0=\rho ((1+r_{0,1})^{-1}\cdot AC_1-AC_0)
\end{equation}
Here, $r_{0,j}$ is the risk-free interest rate at date $i=0$ for a deterministic payoff of $1$ at date $i=j$ in the currency used. The risk measure $\rho$ is a map from a space of random variables $Y$ into $\mathbb{R}\cup\{\infty\}$, with negative realizations of $Y$ corresponding to losses and $\rho (Y)>0$ to positive risk of losses. Solvency II uses the value-at-risk $\rho (Y)=VaR_{\alpha}(Y)=q_{1-\alpha}(-Y)$ for $\alpha =0.005$, with the $u$-quantile $q_u(Y)=\inf\{y\in\mathbb{R}\mid P[Y\leq y]\geq u\}$ (see \cite{Acerbi2022}). The SST uses the lower expected shortfall $\rho (Y)=ES_{\alpha}(Y)=-\frac{1}{\alpha}\cdot\int_{0}^{\alpha}q_u(Y)du$ for $\alpha = 0.01$ and $\mathbb{E}[Y_{-}]<\infty$ (see \cite{Acerbi2022}). Both risk measures are translation-invariant ($\rho (Y+a)=\rho (Y)-a$ for any $a\in\mathbb{R}$) and positively homogeneous ($\rho (a\cdot Y)=a\cdot \rho (Y)$ for any $a\geq 0$).

The value of the insurance liabilities $\mathcal{L}$ at date $i=0$ is given by the sum of best estimate $BEL_0$ and risk margin $RM_0$ (market value margin in the SST):
\begin{equation}
BEL_0+RM_0
\end{equation}
The best estimate $BEL_0$ is roughly defined as the expected value of all discounted cash flows needed for producing the liabilities with the exception of capital cost (and typically without considering fulfillment in only ``sufficiently many cases"). The risk margin $RM_0$ accounts for the capital cost of the required capital $SCR_i$ over the lifetime $i\in\{0,\ldots ,T\}$ of the insurance liabilities, specified through a cost-of-capital rate $\eta >0$ in excess of the risk-free rate. In Solvency II regulation, the following formula for the risk margin is provided, with $\eta =CoC$ (Article 37 of \cite{EUDR2014}): 
\begin{equation}\label{eqn-siiscrformula}
RM_0=CoC\cdot \sum_{i\geq 0} \frac{SCR_i}{(1+r_{0,i+1})^{i+1}}
\end{equation}
The $SCR_i$ are calculated under specific assumptions including no new business, available capital $AC_i$ equal to required capital $SCR_i$, and an investment strategy that is adapted to the insurance liabilities, e.g. under Solvency II so that market risk is ``minimized". 

{\bf Integration into the framework.} For the fulfillment condition, we get from \eqref{eq-avandreqcap} and \eqref{eq-regreqsolvcap} using translation invariance and positive homogeneity of $\rho$ that $AC_0\geq (1+r_{0,1})^{-1}\cdot \rho (AC_1)+AC_0$, which suggests that the fulfillment condition is $\rho (AC_1)\leq 0$. To be more precise, in our setting of production strategies, $AC_1$ here and in \eqref{eq-regreqsolvcap} is, for $i=1$, the available capital $AC_i^{pre}=A_i^{\prime}-L_i$ before liability cash flows $X_{i}^{\mathcal{L}}$, capital payback, and capital raise at date $i$, respectively, using Definition~\ref{def-assetsandliabilities}. In contrast, $AC_0$ above and in \eqref{eq-regreqsolvcap} is, for $i=0$, the available capital $AC_i^{post}=AC_i^{pre}-(A_i^{\prime}-L_i)_{+}+C_i=C_i$ after these steps, in case $A_i^{\prime}\geq L_i$. The value $\phi_{\gamma(i)}\cdot S_{i}$ of the assets in the ``post"-case is the sum of the production cost $\Bar{v}_{i}^{\phi, \psi,\mathcal{C}}(\mathcal{L})$ and the capital $C_i$.

For the financiability condition, we interpret the cost-of-capital rate $\eta >0$ as the minimum (real-world) expected return on the required capital $SCR_0$ in excess of the risk-free return. We thus assume:
\begin{enumerate}[(a)]
\item The fulfillment condition is
\begin{equation}
\rho (AC_1^{pre})\leq 0
\end{equation}
with a translation-invariant and positively homogeneous risk measure $\rho$.  
\item The financiability condition is the requirement that the capital payoff $C_1^{\prime}=(AC_1^{pre})_{+}$ provide at least a (real-world) expected excess return $\eta \geq 0$ over risk free on the capital $C_{0}$, i.e. 
\begin{equation}\label{eqn-siisstfincond}
\mathbb{E}[(AC_1^{pre})_{+}]\geq (1+r_{0,1}+\eta)\cdot C_0
\end{equation}
We write $\mathbb{E}[(AC_1^{pre})_{+}]=\mathbb{E}[1_{M_1} \cdot (A_1^{\prime}-L_1)]$ with $M_1=\{A_1^{\prime}\geq L_1\}$.
\end{enumerate}
The financiability condition \eqref{eqn-siisstfincond} satisfies the two conditions from Definition~\ref{def-financiabilitycondition} and is positively homogeneous but not necessarily consistent with and neutral to the tradables (and may not be a realistic condition). 

We further assume that there are no cash flows for dates $t\in D$ with $i<t<i+1$ and that, for any $i\in\{0,\ldots ,T-1\}$, there exists a ``risk-free one-year zero-coupon bond" with rate $R_{i,i+1}$ from date $i$ to $i+1$ in the selected currency, i.e. a tradable $k\equiv k_i$ with $S_i^{k}=(1+R_{i,i+1})^{-1}$, $S_{i+1}^{k}=0$, and $\tilde{Z}_{i+1}^{k}=Z_{i+1}^{k}=1$. We have $r_{0,1}=R_{0,1}$.

The fulfillment condition for $\rho$ given by the value-at-risk as under Solvency II is equivalent to $P[A_{1}^{\prime}\geq L_{1}]\geq 1-\alpha$ by standard arguments about the infimum. So, the liabilities indeed need to be produced in intuitively ``many cases". For the SST with the expected shortfall, the interpretation is less straightforward. The fulfillment condition is equivalent to $\int_{0}^{\alpha}q_u(AC_1^{pre})du\geq 0$. This implies in particular that $P[A_{1}^{\prime}\geq L_{1}]\geq 1-\alpha^{\prime}>1-\alpha$ for some $\alpha^{\prime}<\alpha$, so the ``safety level" in terms of value-at-risk is higher than $\alpha =0.01$. This follows as $\int_{0}^{\alpha}q_u(AC_1^{pre})du\geq 0$ implies that $\beta = \inf\{u\in\mathbb{R}\mid q_u(AC_1^{pre})\geq 0\} <\alpha$. So, for any $\beta <\alpha^{\prime}<\alpha$, we have $q_{\alpha^{\prime}}(AC_1^{pre})\geq 0$, which implies by standard infimum arguments that $VaR_{\alpha^{\prime}}(AC_1^{pre})=q_{1-\alpha^{\prime}}(-AC_1^{pre})\leq 0$.

To investigate valuation in the framework under the above assumptions, we use the recursively backward-in-time approach of the framework. We consider three stages of successively more simplifications, with each stage moving closer to the usual presentation of Solvency II and SST valuation. In the first stage, the value is calculated as a whole, with no separate calculation of $BEL_0$ and $RM_0$. In the second stage, $BEL_0$ and $RM_0$ are calculated separately, based on restricting the production strategies. We note that the separate calculation creates additional work and would not be required within the framework. In the third stage, the financiability condition \eqref{eqn-siisstfincond} is replaced by a simpler condition.  

To simplify, we assume that we are given the liability value $L_1=X_{1}^{\mathcal{L}}+\Bar{v}_{1}^{\phi, \psi,\mathcal{C}}(\mathcal{L})=X_{1}^{\mathcal{L}}+BEL_1+RM_1$ with best estimate $BEL_1$ and risk margin $RM_1$ at date $i=1$. The objective is to derive the production cost $\Bar{v}_{0}^{\phi, \psi,\mathcal{C}}(\mathcal{L})=BEL_0+RM_0$ at date $i=0$. 

{\bf First stage.} Following the framework, $\Bar{v}_{0}^{\phi, \psi,\mathcal{C}}(\mathcal{L})$ is calculated in the two steps set out in the introduction: 
\begin{enumerate} 
\item[1:] Determine the value $A_0=\phi_{\gamma(0)}\cdot S_{0}$ of the total assets required at date $i=0$ and a production strategy $\phi$ from $i=0$ to $i=1$ such that the fulfillment condition $\rho(AC_1^{pre})\leq 0$ is satisfied at $i=1$.
\end{enumerate} 
In the special case that the assets are only invested in the risk-free one-year zero-coupon bond, i.e. $A_1^{\prime}=(1+r_{0,1})\cdot A_0$, the value $A_0=(1+r_{0,1})^{-1}\cdot\rho(-L_1)$ solves the fulfillment condition with equality, $\rho(A_1^{\prime}-L_1)=0$, as $\rho$ is translation-invariant. We also get $M_1=\{L_1\leq \rho(-L_1)\}$. 

The second step is:
\begin{enumerate} 
\item[2:] Solve for the capital $C_0=SCR_0$ in the financiability condition $(1+r_{0,1}+\eta)\cdot C_0=\mathbb{E}[(AC_1^{pre})_{+}]=\mathbb{E}[1_{M_1} \cdot (A_1^{\prime}-L_1)]$ with equality. From this, we get the production cost $\Bar{v}_{0}^{\phi, \psi,\mathcal{C}}(\mathcal{L})=A_0-SCR_0$ at date $i=0$.
\end{enumerate} 
In the above special case, we have $\mathbb{E}[1_{M_1}\cdot A_1^{\prime}]=(1+r_{0,1})\cdot A_0\cdot \mathbb{E}[1_{M_1}]=(1+r_{0,1})\cdot A_0\cdot P[L_1\leq \rho(-L_1)]$. Hence, for the financiability condition, $(1+r_{0,1}+\eta)\cdot SCR_0=\mathbb{E}[1_{M_1} \cdot (A_1^{\prime}-L_1)]=(1+r_{0,1})\cdot A_0\cdot P[L_1\leq \rho(-L_1)] - \mathbb{E}[1_{\{L_1\leq \rho(-L_1)\}} \cdot L_1]$. Inserting the resulting expression for $SCR_0$ and the above expression $A_0=(1+r_{0,1})^{-1}\cdot\rho(-L_1)$ into $\Bar{v}_{0}^{\phi, \psi,\mathcal{C}}(\mathcal{L})=A_0-SCR_0$, we get an explicit formula for the production cost $\Bar{v}_{0}^{\phi, \psi,\mathcal{C}}(\mathcal{L})$ in terms of $L_1$ (see Theorem 4 of \cite{Moehr2011}).

{\bf Second stage.} We calculate the two components $BEL_0$ and $RM_0$ of the production cost $\Bar{v}_{0}^{\phi, \psi,\mathcal{C}}(\mathcal{L})$, where $RM_0$ should correspond to the capital cost. For this, we make the following two simplifying assumptions:
\begin{enumerate} 
\item[(c)] The assets corresponding to each summand in $A_0=BEL_0+RM_0+ SCR_0$ are invested separately, with the values at date $1$ denoted by $A_1^{\prime}=A_1^{BEL}+A_1^{RM}+A_1^{SCR}$. 
\item[(d)] The assets corresponding to $SCR_0$ and $RM_0$ are invested in the risk-free zero-coupon bond, so $A_1^{SCR}=(1+r_{0,1})\cdot SCR_0$ and $A_1^{RM}=(1+r_{0,1})\cdot RM_0$. 
\end{enumerate} 
Inserting $A_1^{SCR}=(1+r_{0,1})\cdot SCR_0$ into the fulfillment condition with equality $\rho (AC_1^{pre})=0$, using translation-invariance and solving for $SCR_0$, we find that the required capital $SCR_0=(1+r_{0,1})^{-1}\cdot \rho (A_1^{BEL}+A_1^{RM}-L_1)$ arises from the mismatch between the value $A_1^{BEL}+A_1^{RM}$ of the ``replicating portfolio" and the value $L_1$ of the liabilities. However, this formula cannot directly be used to calculate $SCR_0$. Further inserting $A_1^{RM}=(1+r_{0,1})\cdot RM_0$ and using translation invariance: 
\begin{equation}\label{egn-scrrm}
 (1+r_{0,1})\cdot (RM_0+SCR_0) = \rho (A_1^{BEL}-L_1) 
\end{equation}
Inserting this into $A_1^{\prime}=A_1^{BEL} + (1+r_{0,1})\cdot (RM_0 + SCR_0)$, we get $M_1=\{ A_1^{\prime}\geq L_1 \}=\{ A_1^{BEL}-L_1\geq -\rho (A_1^{BEL}-L_1) \}$. 

The financiability condition with equality is that $(1+r_{0,1}+\eta)\cdot SCR_0$ needs to be equal to
\begin{equation}\label{egn-finc2ndstage}
\mathbb{E}[1_{M_1} \cdot (A_1^{BEL}+ A_1^{RM}+A_1^{SCR}-X_{1}^{\mathcal{L}}-BEL_1-RM_1)] 
\end{equation}
To split the value into $BEL_0$ and $RM_0$, considering that the best estimate is intended to account for the expected cash flows other than capital cost in the ``sufficiently many cases", we require $A_1^{BEL}$ to satisfy:
\begin{enumerate} 
\item[(e)] The best estimate is defined by the condition:\footnote{This condition involves the risk margin $RM_1$ through $M_1$, so it is not totally independent of the capital cost.} 
\begin{equation}\label{cond-sepcond}
\mathbb{E}[1_{M_1} \cdot (A_1^{BEL}-X_{1}^{\mathcal{L}}-BEL_1)]=0
\end{equation} 
\end{enumerate} 
This can be solved for $BEL_0$ and its investment strategy from date $0$ to $1$ without knowing $SCR_0$ or $RM_0$ because of the expression for $M_1$ derived above. 

For calculating the risk margin $RM_0$ (and if required $SCR_0$), we insert \eqref{cond-sepcond} into \eqref{egn-finc2ndstage} and derive from the financiability condition that: 
\begin{equation}\label{eqn-rmgeneral}
(1+r_{0,1}+\eta)\cdot SCR_0=P[M_1]\cdot (1+r_{0,1})\cdot (SCR_0+RM_0)-\mathbb{E}[1_{M_1} \cdot RM_1]
\end{equation}
Solving this for the risk margin $RM_0$, it can be written as the sum of a term that can be interpreted as capital cost on the required capital $SCR_0$ for the period from $0$ to $1$ and a term containing $RM_1$, i.e. the capital cost for subsequent periods. To calculate the risk margin $RM_0$, we insert \eqref{egn-scrrm} into \eqref{eqn-rmgeneral} to derive a formula for $SCR_0$. This can then be inserted into \eqref{egn-scrrm} and solved for $RM_0$ to give:
\begin{equation}\label{eqn-rmexplform}
RM_0= \frac{((1+r_{0,1}+\eta)-P[M_1]\cdot (1+r_{0,1}))\cdot \rho (A_1^{BEL}-L_1)}{(1+r_{0,1}+\eta)\cdot (1+r_{0,1})}   +\frac{\mathbb{E}[1_{M_1} \cdot RM_1]}{(1+r_{0,1}+\eta)}
\end{equation}
{\bf Third stage.} We make the further simplification: 
\begin{enumerate} 
\item[(f)] The financiability condition with equality is replaced (by leaving out $(\cdot)_+$) by the stricter condition $\mathbb{E}[AC_1^{pre}]= (1+r_{0,1}+\eta)\cdot SCR_0$. 
\end{enumerate} 
This gives a lower required capital $SCR_0$ and thus a higher production cost than without the simplification and in this sense an upper bound. Formula \eqref{eqn-rmgeneral} then simplifies (as $M_1$ disappears), and solving it for the risk margin $RM_0$:
\begin{equation} 
RM_0=(1+r_{0,1})^{-1}\cdot\eta\cdot SCR_0+(1+r_{0,1})^{-1}\cdot\mathbb{E}[RM_1]
\end{equation} 
So the risk margin is the sum of the discounted capital cost for the period from $0$ to $1$ and the discounted expected capital cost for the periods after $1$, which is a first step for deriving formula \eqref{eqn-siiscrformula} by suitable recursion. As above, $RM_0$ cannot directly be calculated from this but from \eqref{eqn-rmexplform} (with $M_1$ disappearing):
\begin{equation}
RM_0= \frac{\eta\cdot \rho (A_1^{BEL}-L_1)}{(1+r_{0,1}+\eta)\cdot (1+r_{0,1})}   +\frac{\mathbb{E}[RM_1]}{(1+r_{0,1}+\eta)}
\end{equation}
Further assuming that also the assets for the best estimate are invested in the risk-free one-year zero-coupon bond, i.e. $A_{1}^{BEL}=(1+r_{0,1})\cdot BEL_0$, the correspondingly simplified condition \eqref{cond-sepcond} directly gives 
\begin{equation}
BEL_0=(1+r_{0,1})^{-1}\cdot \mathbb{E} [X_{1}^{\mathcal{L}}] + (1+r_{0,1})^{-1}\cdot\mathbb{E} [BEL_{1}] 
\end{equation}
So, the best estimate is the sum of the discounted expected liability cash flows at date $1$ and the discounted expectation of the best estimate at $1$ for liability cash flows after $1$, which for this special case broadly corresponds to Solvency II and the SST.

\section{Conclusion }\label{sec-conclusion}

We have introduced a framework for the valuation of insurance liabilities by production cost for an insurance company subject to regulatory solvency capital requirements and insolvency laws. It considers fulfillment only in ``sufficiently many cases", capital, insolvency, and illiquid assets in addition to tradables, in discrete time. The framework is elementary and in particular does not assume the existence of ``risk-neutral" pricing measures. It is intended to be practically applicable and, through explicitly using production strategies, allows for a concrete fulfillment by production of the insurance liabilities in practice.

The framework assumes that liability cash flows occur and tradables are traded at discrete dates and that insolvency and regulatory solvency capital requirements are assessed annually. It is based on two conditions: a fulfillment condition, defining when fulfillment is ``sufficient", and a financiability condition, defining the return requirements of capital investors. Defining the central notion of production strategies in the general case requires distinguishing between positive shares of tradables, which are assets, and short positions of tradables, which are liabilities, and making explicit the assumptions on taking on and closing out short positions. We define production strategies initially for portfolios with positive shares of tradables and extend it to portfolios with non-negative values. 

We show that, under suitable assumptions on the financiability condition and the applicable insolvency laws, a production strategy for an arbitrary fulfillment condition can be extended to ``full" fulfillment of liabilities that are adjusted only on the complement of the ``sufficiently many cases" and in a sense correspond to a ``redefinition" of the liabilities that is ``allowed" by the given fulfillment condition. We identify the conditions under which the valuation resulting from the production cost can be viewed as an extension of the valuation by market price for investment strategies of tradables. We investigate the implications of the treatment of illiquid assets in the production on the resulting production cost and sketch how the valuation approaches for insurance liabilities from Solvency II and the SST can be seen as special cases of the framework by means of a fulfillment condition that corresponds to the regulatory solvency capital requirement and a specific ``cost-of-capital" fulfillment condition, potentially using additional simplifying assumptions.

\section{Appendix}\label{sec-appendix}

\subsection{Production strategies covering failure - definitions, proposition, proof}\label{ss-prodstratinclfailuredetail}

We assume that according to the applicable insolvency laws, in case of balance sheet insolvency $A^{\prime}_i< L_i$ at a date $i$, all outstanding cash inflows and outflows of the liabilities $\mathcal{L}$ are proportionally reduced according to Definition~\ref{def-adjustedliabilities} by essentially the same factor such that balance sheet insolvency is just removed, leading to adjusted liabilities $\widetilde{\mathcal{L}}^{(i)}$. The adjustment applies inductively forward in time: after liabilities have been reduced at date $i$, failure may under some developments occur again at later dates. So liabilities $\widetilde{\mathcal{L}}^{(i)}$ are successively adjusted ultimately to adjusted liabilities $\widetilde{\mathcal{L}}$. The adjusted liabilities $\widetilde{\mathcal{L}}$ depend on the strategy $\phi$, the illiquid assets $\psi$ and, in our approach, additionally an investment strategy $\theta$, as we explain below. As we show in Proposition~\ref{prop-completeproductionstrategy}, the adjusted liabilities can be produced by suitably scaling the original production strategy $\phi$. The approach is slightly complicated because the illiquid assets $\psi$ cannot just be scaled down as by assumption they cannot be reduced as they cannot be sold. This requires selecting a strategy $\theta\equiv \theta^{\psi}$ in which the excess $(1-\lambda_i)\cdot \tilde{Z}^{\psi}_t$ of the cash inflows $\tilde{Z}^{\psi}_t$ above the scaled down cash flows $\lambda_{i}\cdot \tilde{Z}^{\psi}_t$ is invested. 
\begin{defn}[Adjusted liabilities]\label{def-adjustedliabilities} 
Given a non-negative production strategy $\phi$ with illiquid assets $\psi$ for liabilities $\mathcal{L}$ from $i_{min}$ to $i_{max}$, with the production cost $\Bar{v}_{i}^{\phi, \psi, \mathcal{C}}(\mathcal{L})$ available for all states, the adjusted liabilities $\widetilde{\mathcal{L}}\equiv\widetilde{\mathcal{L}}(\phi,\psi,\mathcal{C},\theta^{\psi})$ are defined for $t\in D\equiv D[i_{min},i_{max}]$ by the cash flows
\begin{equation}
	\tilde{Z}_{t}^{\widetilde{\mathcal{L}}}=\lambda_{\lceil t-1\rceil}\cdot \tilde{Z}_{t}^{\mathcal{L}},\;\;\; X_{t}^{\widetilde{\mathcal{L}}}=\lambda_{\lfloor t\rfloor}\cdot X_{t}^{\mathcal{L}}
\end{equation}
where $\lceil t-1\rceil=i-1$ for $t=i$ and $\lceil t-1\rceil=i$ for $i< t\leq i+1$, and $\lfloor t\rfloor=i$ for $i\leq t<i+1$. The adapted random variables $(\lambda_i)_{i=i_{min}}^{i_{max}}$ with $0\leq\lambda_i\leq 1$ are defined recursively by $\lambda_{i}=\xi_i\cdot \lambda_{i-1}$ with $\lambda_{i_{min}-1}=1$ and 
\begin{equation}\label{defn-xii}
	\xi_i=\frac{(\lambda_{i-1}\cdot L_i)\wedge (\lambda_{i-1}\cdot A^{\prime}_i+X^{\theta}_i)}{\lambda_{i-1}\cdot L_i}\mbox{ for }L_i>0\mbox{, and }\xi_i=1\mbox{ for }L_i=0
\end{equation}
The  non-negative investment strategy $\theta\equiv \theta^{\psi}$ in \eqref{defn-xii} is arbitrary apart from its cash flows given by, for $t\in D$ and $i\in\{i_{min},i_{min}+1,\ldots, i_{max}\}$,
\begin{equation}
	\tilde{Z}^{\theta}_t = (1-\lambda_{\lceil t-1\rceil})\cdot \tilde{Z}^{\psi}_t,\;\; X^{\theta}_i = \theta_i\cdot S_i +(1-\lambda_{i-1})\cdot \tilde{Z}^{\psi}_i\mbox{ and zero otherwise} 
\end{equation}
where the cash flows $\tilde{Z}^{\theta}_t$ are incoming to $\theta$ from $\psi$ and $X^{\theta}_i$ are outgoing of $\theta$ but remain within the company's assets. Such a strategy can always be defined as, at any date $i$, it only pays out its value and the cash inflows.
\end{defn}
Note that, at date $t = i$, cash inflows of the liability $\mathcal{L}$ are scaled by $\lambda_{i-1}$ and cash outflows by $\lambda_i$, which means that 
 premiums at date $i$ are assumed to be considered in full before failure is declared, whereas outgoing claims, benefits and costs are subsequently reduced, including those at date $i$. For the remaining dates $t\in D\setminus\{0,1,\ldots , T\}$, the same scaling factors are used. In Definition~\ref{def-adjustedliabilities}, liability cash flows are only scaled down over time, although the situation may improve in some developments. 
\begin{prop}\label{prop-completeproductionstrategy}
Let the financiability condition be positively homogeneous. Let $\phi\in\mathcal{R}^{\geq 0}$ with illiquid assets $\psi$ and capital amounts $\mathcal{C}=(C_i)_{i=i_{min}}^{i_{max}-1}$  be a non-negative production strategy from $i_{min}$ to $i_{max}$ for the liability $\mathcal{L}$ with terminal value $Y_{i_{max}}$. Assume that the production cost $\Bar{v}_{i}^{\phi, \psi,\mathcal{C}}(\mathcal{L})$ for all $i\in\{i_{min},\ldots ,i_{max}-1\}$ are available for all states. Given any fixed non-negative strategy $\theta\equiv \theta^{\psi}$ as in Definition~\ref{def-adjustedliabilities}, define the strategy $\widetilde{\phi}\in\mathcal{R}^{\geq 0}$ for $t\in D[i_{min}, i_{max}]$ by 
\begin{equation}
	\widetilde{\phi}_t=\lambda_{\lceil t-1\rceil}\cdot \phi_t \mbox{ and including }X^{\theta}_t\mbox{ as cash inflows}
\end{equation}
with illiquid assets $\widetilde{\psi}_t=\lambda_{\lceil t-1\rceil}\cdot\psi_t$, capital amounts $\widetilde{\mathcal{C}}=(\lambda_{i}\cdot C_i)_{i=i_{min}}^{i_{max}-1}$ and terminal value $\tilde{Y}_{i_{max}}=\lambda_{i_{max}}\cdot Y_{i_{max}}$. Then, the strategy $\widetilde{\phi}\in\mathcal{R}^{\geq 0}$ extends $\phi$ to a production strategy for the full fulfillment condition and the same financiability condition for the liabilities $\widetilde{\mathcal{L}}\equiv \widetilde{\mathcal{L}}(\phi,\psi,\mathcal{C},\theta)$ and has production cost
\begin{equation}
 \Bar{v}_{i}^{\tilde{\phi},\psi,\widetilde{\mathcal{C}}}(\widetilde{\mathcal{L}})=\lambda_{i}\cdot \Bar{v}_{i}^{\phi,\psi,\mathcal{C}}(\mathcal{L}) \;\;\mbox{ for }i\in\{i_{min},\ldots,i_{max}-1\}  
\end{equation}
\end{prop}
In the proposition, $\widetilde{\phi}_t$ being defined as ``including $X^{\theta}_t$ as cash inflows" means that, at a date $i$ ($X^{\theta}_t$ is zero at other dates), the cash flows $X^{\theta}_t$ flow into the strategy. This does not change the strategy because the additional cash flows are immediately paid out to the capital investors through $C_{i}^{\prime}$. In that sense, the capital investors potentially receive ``too much". 
\begin{proof}
The adjusted liabilities $\widetilde{\mathcal{L}}$ and the strategy $\widetilde{\phi}$ are defined inductively for increasing $i\in \{i_{min}-1\}\cup I$ with $I\equiv\{i_{min},\ldots ,i_{max}\}$ by adjusted liabilities $\widetilde{\mathcal{L}}^{(i)}$, strategies $\widetilde{\phi}^{(i)}$, illiquid assets $\psi^{(i)}_t$, suitable capital amounts $C_j^{(i)}$, and strategies $\theta^{(i)}_t$, for a suitable terminal value $\tilde{Y}_{i_{max}}^{(i)}$. The quantities in step $i$ account for failure - by adjusting the liabilities such that the full fulfillment condition is satisfied - up to and including date $i$, but potentially not at later dates. This is achieved by scaling all quantities at step $i$ instead of by $\lambda_j$ by scaling factors $\lambda_j^{(i)}$ for $j\in I$ that are equal to $\lambda_j$ for $j\leq i-1$ and equal to $\lambda_i$ for $j\geq i$ and with $\lambda_{i_{min}-1}^{(i)}=\lambda_{i_{min}-1}=1=\lambda_{j}^{(i_{min}-1)}$. The strategy $\theta_t^{(i)}$ satisfies $\theta_t^{(i)}=\theta_t$ for $t\leq i+1$ and for the values of the assets and liabilities, $A^{\prime (i)}_{j}=\lambda_{j-1}^{(i)}\cdot A^{\prime}_{j}$ and $L_{j}^{(i)}=\lambda_{j}^{(i)}\cdot L_{j}$ for $j\in I$. In the iteration step from date $i$ to $i+1$, accounting for failure is extended to date $i+1$ without changing the situation up to date $i$. To this end, we define $\xi_j^{(i+1)}$ to be $\xi_j$ for $j\leq i$,
\begin{equation}\label{def-xiip1}
\xi_{i+1}^{(i+1)}=\frac{L_{i+1}^{(i)}\wedge (A^{\prime (i)}_{i+1}+X^{\theta^{(i)}}_{i+1})}{L_{i+1}^{(i)}}
\end{equation}
and equal to $1$ for $j\geq i+2$. We have $\xi_{i+1}^{(i+1)}=\xi_{i+1}$ because $X_{i+1}^{\theta^{(i)}}=\theta_{i+1}^{(i)}\cdot S_{i+1} +(1-\lambda_{i}^{(i)})\cdot \tilde{Z}^{\psi}_{i+1}=\theta_{i+1}\cdot S_{i+1} +(1-\lambda_{i})\cdot \tilde{Z}^{\psi}_{i+1}=X_{i+1}^{\theta}$, and $A^{\prime (i)}_{i+1}=\lambda_{i}\cdot A^{\prime}_{i+1}$, and $L^{(i)}_{i+1}=\lambda_{i}\cdot L_{i+1}$. Thus, $\xi_j^{(i+1)}$ is equal to $\xi_j$ for $j\leq i+1$ and equal to $1$ for $j\geq i+2$. In particular, the situation is unchanged from date $i$ in that $\xi_j^{(i+1)}=\xi_j^{(i)}$ for $j\leq i$. We define the scaling factors $\lambda_{i_{min}-1}^{(i+1)}=1$ and $\lambda_{j}^{(i+1)}=\xi_{j}^{(i+1)}\cdot\lambda_{j-1}^{(i+1)}$ for $j\in I$, which implies: 
\begin{equation}
\lambda_j^{(i+1)} = \left\{ 
\begin{array}{rcl} 
& \lambda_j & \mbox{ for } j\leq i \\ 
& \lambda_{i+1} & \mbox{ for } j\geq i+1
\end{array}
\right.
\end{equation}
In particular, $\lambda_j^{(i+1)}=\lambda_j^{(i)}$ for $j\leq i$. The adjusted liabilities $\mathcal{L}^{(i+1)}$ are defined by scaling the cash flows analogously to Definition~\ref{def-adjustedliabilities}. To define the adjusted strategy $\tilde{\phi}^{(i+1)}$ with $\psi^{(i+1)}_t$, we define $\theta^{(i+1)}$ by $\tilde{Z}^{\theta^{(i+1)}}_t = (1-\lambda_{\lceil t-1\rceil}^{(i+1)})\cdot \tilde{Z}^{\psi}_t$ and $X^{\theta^{(i+1)}}_j = \theta_j^{(i+1)}\cdot S_j +(1-\lambda_{j-1}^{(i+1)})\cdot \tilde{Z}^{\psi}_j$ and investing in the same way as $\theta^{(i)}_t$ for $t\leq i+1$. It follows that $\theta^{(i+1)}_t=\theta_t$ for $t\leq i+2$ and $\theta^{(i+1)}_t=\theta^{(i)}_t$ for $t\leq i+1$. We can then define $A_j^{\prime (i+1)}=\tilde{\phi}_j^{(i+1)}\cdot S_j+\tilde{Z}_j^{\tilde{\phi}^{(i+1)}+\psi^{(i+1)}+\tilde{\mathcal{L}}^{(i+1)}}=\lambda_{j-1}^{(i+1)}\cdot \phi_j\cdot S_j+\lambda_{j-1}^{(i+1)}\cdot\tilde{Z}_j^{\phi+\psi+\mathcal{L}}$. Hence, $A_j^{\prime (i+1)}=\lambda_{j-1}^{(i+1)}\cdot A_j^{\prime}$ is equal to $\lambda_{j-1}\cdot A_j^{\prime}$ for $j\leq i+1$ and $\lambda_{i+1}\cdot A_j^{\prime}$ for $j\geq i+2$. In particular, $A_j^{\prime (i+1)}=A_j^{\prime (i)}$ for $j\leq i$. We define 
\begin{equation}
L_j^{(i+1)}=X_j^{\tilde{\mathcal{L}}^{(i+1)}}+\Bar{v}_{j}^{\tilde{\phi}^{(i+1)},\psi^{(i+1)},\mathcal{C}^{(i+1)}}(\tilde{\mathcal{L}}^{(i+1)})
\end{equation}
with $X_j^{\tilde{\mathcal{L}}^{(i+1)}}=\lambda_j^{(i+1)}\cdot X_j^{\mathcal{L}}$. We define $C_j^{\prime (i+1)}=(A_j^{\prime (i+1)}+X_j^{\theta^{(i+1)}}-L_j^{(i+1)})_{+}$, and $C_j^{(i+1)}=\lambda_{j}^{(i+1)}\cdot C_j$, and $Y_{i_{max}}^{(i+1)}=\lambda_{i_{max}}^{(i+1)}\cdot Y_{i_{max}}$. 

We show that $\tilde{\phi}^{(i+1)}$ is a production strategy for the liabilities $\widetilde{\mathcal{L}}^{(i+1)}$. We have $\tilde{\phi}^{(i+1)}\in \mathcal{R}^{\geq 0}$ as it is a non-negative multiple of $\phi\in\mathcal{R}^{\geq 0}$. To show for $\tilde{\phi}^{(i+1)}$ the equation corresponding to \eqref{def-generalfinancing} for $t\in D[i_{min},i_{max}]$ with $t>i$ and $t\notin \{i+1,i+2,\ldots\}$, we use $X^{\theta^{(i+1)}}_{t}=0$ and that for $s\in\{ t, \gamma (t)\}$ by definition, $\lambda_{\lceil s-1\rceil}^{(i+1)}=\lambda_{\lfloor s\rfloor}^{(i+1)}=\lambda_{i+1}$. So, each of the two sides of \eqref{def-generalfinancing} for $\tilde{\phi}^{(i+1)}$ is equal to the corresponding side for $\phi$ multiplied by $\lambda_{i+1}$, so \eqref{def-generalfinancing} for $\tilde{\phi}^{(i+1)}$ follows from \eqref{def-generalfinancing} for $\phi$.

For the production cost of $\tilde{\phi}^{(i+1)}$, we show recursively backward in time for $j\in I$ that 
\begin{equation}\label{eqn-valuerecursionadj}
\Bar{v}_{j}^{\tilde{\phi}^{(i+1)},\psi^{(i+1)},\mathcal{C}^{(i+1)}}(\tilde{\mathcal{L}}^{(i+1)})=\lambda_j^{(i+1)}\cdot\Bar{v}_{j}^{\phi,\psi,\mathcal{C}}(\mathcal{L})
\end{equation}
This then implies in particular that $L_j^{(i+1)}=\lambda_j^{(i+1)}\cdot L_j$. 

Equation \eqref{eqn-valuerecursionadj} holds for $j=i_{max}$ by definition of $Y_{i_{max}}^{(i+1)}$. Let $i+1\leq j\leq i_{max}-1$. Then, $\lambda_{j+1}^{(i+1)}=\lambda_{i+1}=\lambda_{j}^{(i+1)}$. Given that \eqref{eqn-valuerecursionadj} holds for date $j+1$ by recursion, we get  $L_{j+1}^{(i+1)}=\lambda_{j+1}^{(i+1)}\cdot L_{j+1}=\lambda_j^{(i+1)}\cdot L_{j+1}$. Thus, if $A_{j+1}^{\prime}\geq L_{j+1}$, then, as $\theta^{i+1}$ is non-negative, $A_{j+1}^{\prime (i+1)}+X_{i+1}^{\theta^{(i+1)}}\geq A_{j+1}^{\prime (i+1)}=\lambda_{j}^{(i+1)}\cdot A_{j+1}^{\prime}\geq  \lambda_{j}^{(i+1)}\cdot L_{j+1}=L_{j+1}^{(i+1)}$, so the original fulfillment condition holds for the adjusted liabilities. Also, $C_{j+1}^{\prime (i+1)}\geq (A_{j+1}^{\prime (i+1)}-L_{j+1}^{(i+1)})_{+}=\lambda_{j}^{(i+1)}\cdot C_{j+1}^{\prime}$. As the financiability condition is assumed to be positively homogeneous, it holds for $C_j^{(i+1)}=\lambda_{j}^{(i+1)}\cdot C_j$. Hence, \eqref{eqn-valuerecursionadj} holds for date $j$ as $\Bar{v}_{j}^{\tilde{\phi}^{(i+1)},\psi^{(i+1)},\mathcal{C}^{(i+1)}}(\tilde{\mathcal{L}}^{(i+1)})=\tilde{\phi}_{\gamma (j)}^{(i+1)}\cdot S_j-C_j^{(i+1)}=\lambda_{j}^{(i+1)}\cdot \Bar{v}_{j}^{\phi,\psi,\mathcal{C}}(\mathcal{L})$, and the recursion proceeds.

Now assume that \eqref{eqn-valuerecursionadj} holds at date $j=i+1$. This implies $L_{i+1}^{(i+1)}=\lambda_{i+1}^{(i+1)}\cdot L_{i+1}=\lambda_{i+1}\cdot L_{i+1}=\xi_{i+1}\cdot \lambda_{i}\cdot L_{i+1}=(\lambda_{i}\cdot L_{i+1})\wedge (\lambda_{i}\cdot A_{i+1}^{\prime}+X_{i+1}^{\theta})$. So, $L_{i+1}^{(i+1)}\leq \lambda_{i}\cdot A_{i+1}^{\prime}+X_{i+1}^{\theta}=A_{i+1}^{\prime (i+1)}+X_{i+1}^{\theta^{(i+1)}}$, that is, the full fulfillment condition holds at date $i+1$. We have $C_{i+1}^{\prime (i+1)}\geq (A_{i+1}^{\prime (i+1)}-L_{i+1}^{(i+1)})_{+}=(\lambda_{i}^{(i+1)}\cdot A_{i+1}-\lambda_{i+1}^{(i+1)}\cdot L_{i+1})_{+}$. From $\xi_{i+1}^{(i+1)}\leq 1$ we get $\lambda_{i+1}^{(i+1)}\leq \lambda_{i}^{(i+1)}$, hence $C_{i+1}^{\prime (i+1)}\geq \lambda_{i}^{(i+1)}\cdot (A_{i+1}-L_{i+1})_{+}=\lambda_{i}^{(i+1)}\cdot C_{i+1}^{\prime}$. So we can proceed as above to show that \eqref{eqn-valuerecursionadj} holds at date $i$. 

Finally, for dates $j\leq i$, the situation is unchanged from the previous iteration with $i$ and so \eqref{eqn-valuerecursionadj} also holds at these dates. This completes the iteration step from date $i$ to $i+1$.

\end{proof}

\newpage

\nocite{Foellmer2016}
\nocite{Engsner2017}
%\nocite{Engsner2021}

\bibliographystyle{plain}

%\bibliography{Biblio}

\end{document}